\documentclass[english]{article}
\usepackage{geometry}
\usepackage{float}
\usepackage{mathrsfs}
\usepackage{amsmath}
\usepackage{amssymb}
\usepackage{graphicx}
\usepackage{amsfonts}
\usepackage{bm}
\usepackage{subfigure}
\usepackage{amsthm}
\usepackage{epsfig}
\makeatletter

\@ifundefined{definecolor}{\@ifundefined{definecolor}
 {\@ifundefined{definecolor}
 {\usepackage{color}}{}
}{}
}{}

\usepackage{subfig}\usepackage[all]{xy}

\newcommand{\bbR}{\mathbb R}

\newtheorem{theorem}{Theorem}[section]
\newtheorem{lem}{Lemma}[section]
\newtheorem{rem}{Remark}[section]
\newtheorem{prop}{Proposition}[section]

\newcounter{hypA}
\newenvironment{hypA}{\refstepcounter{hypA}\begin{itemize}
  \item[({\bf A\arabic{hypA}})]}{\end{itemize}}

\usepackage{babel}\date{}

\begin{document}

\begin{center}

{\Large \textbf{A Stable Particle Filter in High-Dimensions}}

\bigskip

BY  ALEXANDROS BESKOS, DAN CRISAN, AJAY JASRA, KENGO KAMATANI, \& YAN ZHOU

{\footnotesize Department of Statistical Science,
University College London, London, WC1E 6BT, UK.}
{\footnotesize E-Mail:\,}\texttt{\emph{\footnotesize a.beskos@ucl.ac.uk}}\\
{\footnotesize Department of Mathematics,
Imperial College London, London, SW7 2AZ, UK.}\\
{\footnotesize E-Mail:\,}\texttt{\emph{\footnotesize d.crisan@ic.ac.uk}}\\
{\footnotesize Department of Statistics \& Applied Probability,
National University of Singapore, Singapore, 117546, SG.}\\
{\footnotesize E-Mail:\,}\texttt{\emph{\footnotesize staja@nus.edu.sg, stazhou@nus.edu.sg}}\\
{\footnotesize Graduate School of Engineering Science, Osaka University, Osaka, 565-0871, JP.}
{\footnotesize E-Mail:\,}\texttt{\emph{\footnotesize kamatani@sigmath.es.osaka-u.ac.jp}}\\
\end{center}

\begin{abstract}
We consider the numerical approximation of the filtering problem in high dimensions, that is, when the hidden state lies in $\mathbb{R}^d$ with $d$ large.
For low dimensional problems, one of the most popular numerical procedures for consistent inference is the class of approximations termed   particle filters or sequential Monte Carlo methods. 
However, in high dimensions, standard particle filters (e.g. the bootstrap particle filter)  can have a cost that is exponential in $d$
for the algorithm to be stable in an appropriate sense. We develop a new particle filter, called the \emph{space-time particle filter},  for a specific family of state-space models in discrete time.
This new class of particle filters provide consistent Monte Carlo estimates for any fixed $d$, as do standard particle filters. Moreover, 
we expect that the state-space particle filter will scale much better 
with $d$ than the standard filter. We illustrate this analytically 
for a model of a simple i.i.d.\@ structure and one of a Markovian structure 
in the $d$-dimensional space-direction, when we show that the algorithm 
exhibits certain stability properties as $d$ increases at a cost $\mathcal{O}(nNd^2)$, where $n$ is the time parameter
and $N$ is the number of Monte Carlo samples, that are fixed and independent of $d$. Similar results are expected to hold, under a more general structure than the i.i.d.~one. 
Our theoretical results are also supported by numerical simulations on 
practical models of complex structures. The results suggest that it is indeed possible to tackle some high dimensional filtering problems
using the space-time particle filter that standard particle filters cannot handle.
\end{abstract}

\section{Introduction}\label{sec:intro}

We consider the numerical resolution of filtering problems  
and the estimation of the associated normalizing constants 
for state-space models. In particular, the data is modelled by a discrete time process  $\{Y_n\}_{n\geq 1}$, $Y_n\in\mathbb{R}^{d_y}$, associated to a  hidden signal modelled by a 
Markov chain $\{X_n\}_{n\geq 0}$, $X_n\in\mathbb{R}^d$; we concerned with high dimensions, i.e.~$d$ large. For simplicity, we assume that the location of the signal at time 0 is fixed and known, but the algorithm 
can easily be extended to the general case\footnote{Both the results and the arguments can be extended to unknown initial locations of the signal, i.e., to $X_0$ being a random variable. In this case we require a mechanism through which we can produce a sample from its distribution with a polynomial computational effort in the dimension of the state space.}. We will write the joint density (with respect to an appropriate dominating measure) of $(x_{1:n},y_{1:n})$
as
$$
p(x_{1:n},y_{1:n}) = \prod_{k=1}^n g(x_k,y_k)f(x_{k-1},x_{k}),
$$
for kernel functions $f,g$ and $X_0=x_0$ 
so that,  given the hidden states  $X_{1:n}=\{X_1,...,X_n\}$,
the data $Y_{1:n}=\{Y_1,...,Y_n\}$ consist of independent entries  with $Y_k$ only depending on $X_k$.
The objective is to approximate the filtering distribution $X_n|Y_{1:n}=y_{1:n}$. 
This filtering problem when $d$ is large is notoriously difficult, in many scenarios.

In general, the filter cannot be computed exactly and one often has to resort to numerical methods, for example by using particle filters (see e.g.~\cite{doucet}). Particle filters   make use of  a sequence of proposal densities and sequentially simulate from these a collection of $N>1$ samples, termed particles. In most scenarios
it is not possible to use the distribution of interest as a proposal. Therefore, one must correct for the discrepancy between proposal
and target  via importance weights. In the majority of cases of practical interest, the variance of these importance weights increases with algorithmic time. This can, to some extent, be dealt with via a  resampling procedure consisted of sampling with replacement from the current weighted samples and resetting them to $1/N$.
The variability of the weights is often measured by the effective sample size (ESS).  If $d$ is small to moderate, then particle filters can many times perform very well in the time parameter $n$ (e.g.~\cite{delmoral}). For instance, under conditions the Monte Carlo error of the estimate of the filter can be uniform with respect to the time parameter.

For some state-space models, with specific structures, particle algorithms can work well in high dimensions, or at least can be appropriately modified to do so.
We note for instance that one can set-up an effective particle filter even when $d=\infty$
 provided one assumes a finite (and small, relatively to $d$) amount of information in the likelihood (see e.g.\@ \cite{kantas} for details).  This is \emph{not} the class of problems for which we are interested in here.
In general, it is mainly the amount of information 
in the likelihood $g(x_{k},y_k)$ that determines the algorithmic challenge  rather than the dimension $d$ of the hidden space per-se (this is related to what is called 
`effective dimension' in \cite{bickel}). The function  $x_{k}\mapsto g(x_{k},y_k)$ can convey a lot of information about the hidden state, especially so in high dimensions.
 If this is the case, using the prior transition kernel 
$f(x_{k-1},x_{k})$ as proposal will be ineffective. We concentrate here on the challenging class of problems with large state  space dimension  $d$
and an amount of information in the likelihood 
that increases with $d$.
 It is then known that the standard particle filter will typically perform poorly 
in this context,
often requiring that $N=\mathcal{O}(\kappa^d)$, for some $\kappa>1$, see for instance \cite{bickel}. 
%
%
The results of \cite{bickel}, amongst others, has motivated substantial research in the literature
on particle filters in high-dimensions,  such as the recent work in \cite{rebs}
which attempts an approximate split of the $d$-dimensional state vector to confront 
the curse-of-dimensionality for importance sampling, at the cost of introducing 
difficult to quantify bias with magnitude that depends on the position along the $d$ co-ordinates. See \cite{rebs} and the references therein for some algorithms designed for high-dimensional filtering. To-date, there are few particle filtering algorithms that are:
\begin{enumerate}
\item{asymptotically consistent (as $N$ grows),}
\item{of fixed computational cost per time step (`online'),} 
\item{supported  by theoretical analysis demonstrating  a sub-exponential cost in $d$.}
\end{enumerate}
In this article we attempt to provide an algorithm which has the above properties. 

Our method develops as follows. In a general setting, we assume 
there exists  an increasing sequence of sets 
$\{\mathcal{A}_{k,j}\}_{j=1}^{\tau_{k,d}}$, 
with $\mathcal{A}_{k,1}\subset \mathcal{A}_{k,2}\subset \cdots \subset  \mathcal{A}_{k,\tau_{k,d}}= \{1:d\}$, for some integer $0<\tau_{k,d}\le d$, such that we can 
factorize:
\begin{equation}
g(x_k,y_k)f(x_{k-1},x_{k}) = \prod_{j=1}^{\tau_{k,d}} \alpha_{k,j}(y_k,x_{k-1}, x_{k}(\mathcal{A}_{k,j})), \label{eq:hmm_struc}
\end{equation}
%
%
for appropriate functions $\alpha_{k,j}(\cdot)$,
where we denote $x_{k}(\mathcal{A}) = \{x_{k}(j):j\in \mathcal{A}\}\in\mathbb{R}^{|\mathcal{A}|}$. 
As we will remark later on, this structure is not an absolutely necessary 
requirement for the subsequent
algorithm, but will clarify the ideas in the development of the method. 
%
%
Within a sequential Monte Carlo context, 
one can think of augmenting the sequence of distributions of increasing dimension $X_{1:k}|Y_{1:k}$, $1\le k \le n$, moving from $\mathbb{R}^{d(k-1)}$ to 
$\mathbb{R}^{dk}$, with intermediate laws on $\mathbb{R}^{d(k-1)+
|\mathcal{A}_{k,j}|}$,
for $j=1, \ldots, \tau_{k,d}$.
The structure in (\ref{eq:hmm_struc}) is not uncommon. 
For instance one should typically be able to obtain such a factorization for the 
prior term $f(x_{k-1},x_k)$ by marginalising over subsets 
of co-ordinates. Then, for the likelihood component $g(x_{k},y_k)$ 
this could for instance be implied when the model assumes a local dependence 
structure for the observations. 
Critically, for this approach to be effective
 it is necessary that the factorisation is 
such that will allow for a gradual introduction of the `full' likelihood term $g(x_k,y_k)$ 
along the $\tau_{k,d}$ steps. For instance, trivial choices like  
$\alpha_{k,j} = \int f(x_{k-1},x_k)dx_k(j+1:d)/\int f(x_{k-1},x_k)dx_k(j:d)$, $1\le j\le d-1$, and $\alpha_{k,d}= \Big(f(x_{k-1},x_k)/\int f(x_{k-1},x_k)dx_k(d)\Big)g(x_k,y_k)$ will be ineffective, as they only introduce the complete likelihood term in the last step.

Our contribution is based upon the idea 
that particle filters in general  work well with regards to the time parameter (they are sequential).
Thus, we will exploit the structure in (\ref{eq:hmm_struc}) to build up a
particle filter in space-time 
moving vertically along the space index; for this reason, we call the new algorithm the space-time particle filter (STPF).
 We break the $k$-th time-step  of the particle filter into
$\tau_{k,d}$ space-steps and run a system of $N$ independent particle filters for these steps.  This is similar to a tempering approach as the one 
in \cite{beskos,beskos1}, in the context of sequential Monte Carlo algorithms \cite{delm:06} for a single target probability of dimension $d$. There, the idea is to use annealing steps, interpolating between an easy to sample
distribution and the target with an $\mathcal{O}(d)$ number of steps. In the context of filtering, for the filter, say, at time 1 we break the problem of
trying to perform importance sampling in one step for a $d$-dimensional object (which typically does not perform well, as noted by \cite{bickel}) into $\tau_{1,d}$ easier steps via the particle filter along space; as the particle filter on low to moderate dimensions is
typically well behaved, one expects the proposed procedure to work well even if $d$ is large. A similar idea is used at subsequent time steps of the filter.

In the main part of the paper and in all theoretical derivations, we work under the easier to present scenario 
$\tau_{k,d}=d$ and $\mathcal{A}_{k,j}=\{1:j\}$. 
We establish that our algorithm is consistent as $N$ grows (for fixed $d$), i.e.~that one can estimate the filter with enough computational power, in a manner that is online. The we look at two simple models: a) an i.i.d.\@ scenario both in space and time, b) a Markovian model along space. In both cases, we present results indicating
that the algorithm is stable 
at a cost of $\mathcal{O}(nNd^2)$. 
%
As we remark later on, we expect this cost to be optimistic, but, we conjecture that the cost in general is no worse than polynomial in $d$. 
These claims are further supported by numerical simulations.
We stress here that there is a lot more to be investigated in terms 
of the analytical properties of the proposed algorithm to fully explore its potential,
certainly in more complex model structures than the above.
This work aims to make an important first contribution in an very significant
and challenging problem and open up several 
directions for future investigation.

This article is structured as follows. In Section \ref{sec:algo} the STPF algorithm is given. In Section \ref{sec:theory} our mathematical results are given; some proofs are housed in the Appendix.
In Section \ref{sec:numerics} our algorithm is implemented and compared to existing methodology. In Section \ref{sec:summary} the article is concluded with several remarks for future work.

\section{The Space-Time Particle Filter}\label{sec:algo}

We develop an algorithm that combines a local filter running $d$ space-step using $M_d$ particles, with a global filter making  time-steps and uses $N$ particles. 
We will establish in Section \ref{sec:theory}, that for any fixed $M_d\geq 1, d\geq 1$, the algorithm is
consistent, with respect to some estimates of interest, as $N$ grows. 
A motivation for using such an approach is that it can potentially  provide good estimates 
for expectations over the complete $d$-dimensional filtering density $X_n|Y_{1:n}=y_{1:n}$, 
whereas a standard filter with $N=1$ could exhibit path degeneracy even within 
a single time-step (for large $d$), thus providing unreliable estimates for $X_n|Y_{1:n}=y_{1:n}$.
This approach has been motivated by the island particle model of \cite{verge},
where a related method for standard particle filters (and not related with confronting the dimensionality issue) was developed,
but is not a trivial extension of it, so some extra effort is required to ensure correctness of the algorithm.
We will also explain how to set $M_d$
as a function of $d$ to ensure some stability properties  with respect to $d$ in some specific modelling scenarios.  
The notation $x_n^{i,l}(1:j)\in\mathbb{R}^j$ is adopted, with $i\in\{1,\dots,N\}$, denoting the particle, $n\geq 1$ the
discrete observation time, $1:j$ denoting dimensions $1,\dots,j$ and $l\in\{1,\dots,M_d\}$ the particle in the local system.

\subsection{Time-Step 1}
For each $i\in\{1,\dots,N\}$, the following algorithm is run. We introduce a sequence of proposal densities $q_{1,j}(x_{1}^{i,l}(j)|x_{1}^{i,l}(1:j-1),x_0)$ and will run a particle filter in space-direction that
builds up the dimension towards $x_1\in\mathbb{R}^d$. At space-step 1, one generates $M_d$-samples from $q_{1,1}$ in $\mathbb{R}$ and computes the weights
$$
G_{1,1}(x_{1}^{i,l}(1)) = \frac{\alpha_{1,1}(y_1,x_0, x_{1}^{i,l}(1))}{q_{1,1}(x_{1}^{i,l}(1)|x_0)}, \quad l\in\{1,\dots,M_d\}.
$$
The $M_d$-samples are resampled, according to their corresponding  weights. For simplicity, we will assume   we use multinomial resampling. The resampled particles are
written as $\check{x}_{1}^{i,l}(1)$.
At subsequent points $j\in\{2,\dots,d\}$  one generates $M_d$-samples from $q_{1,j}$ in $\mathbb{R}$ and computes  
$$
G_{1,j}(\check{x}_{1}^{i,l}(1:j-1),x_{1}^{i,l}(j)) = \frac{\alpha_{1,j}(y_1,x_0,\check{x}_{1}^{i,l}(1:j-1),x_{1}^{i,l}(j))}{q_{1,j}(x_{1}^{i,l}(j)|x_0,\check{x}_{1}^{i,l}(1:j-1))}, \quad l\in\{1,\dots,M_d\}.
$$
The $M_d$-samples are resampled  according to the weights. 
At the end of the 1st time-step, all the last particles are resampled, 
thus giving $\check{x}_{1}^{i,l}(1:d)$ (so that we have $N$ independent  particle systems of $M_d$ particles). The $N$ particle systems are assigned weights
\begin{equation}
\mathbf{G}_1(\check{x}_{1}^{i,1:M_d}(1:d-1),x_{1}^{i,1:M_d}(1:d)) = \prod_{j=1}^d \Big(\frac{1}{M_d}\sum_{l=1}^{M_d}G_{1,j}(\check{x}_{1}^{i,l}(1:j-1),x_{1}^{i,l}(j))\Big).
\label{eq:global_g}
\end{equation}
We then resample the $N$-particle systems according to these weights. The normalizing constant $\int_{\mathbb{R}^d}g(x_1,y_1)f(x_{0},x_1)dx_1$ can be estimated by
\begin{equation}
\frac{1}{N}\sum_{i=1}^N\mathbf{G}_1(\check{x}_{1}^{i,1:M_d}(1:d-1),x_{1}^{i,1:M_d}(1:d))\label{eq:nc1}.
\end{equation}
For $\varphi:\mathbb{R}^d\rightarrow\mathbb{R}$, the filter at time 1, 
$$
\frac{\int_{\mathbb{R}^d}\varphi(x_1)g(x_1,y_1)f(x_0,x_{1})dx_1}{\int_{\mathbb{R}^d}g(x_1,y_1)f(x_0,x_{1})dx_1}
$$ 
can be estimated by
\begin{equation}
\frac{1}{NM_d}\sum_{l=1}^{M_{d}}\sum_{i=1}^N\varphi(\check{x}_{1}^{i,l}(1:d))\label{eq:filt1}
\end{equation}
where, with some abuse of notation, we assume that $\check{x}_{1}^{i,l}(1:d)$ have been resampled according to the weights of the global filter in \eqref{eq:global_g}.
We will remark on these estimates later on.

\subsection{Time-Steps $n\geq 2$}

For each $i\in\{1,\dots,N\}$, the following algorithm is run. Introduce a sequence of proposal densities $q_{n,j}(x_{n}^{i,l}(j)|\check{x}_{n}^{i,l}(1:j-1),\check{x}_{n-1}^{i,l}(1:d))$. At step 1, one produces $M_d$-samples from $q_{n,1}$ in $\mathbb{R}$ and computes the weights
$$
G_{n,1}(\check{x}_{n-1}^{i,l}(1:d),x_{n}^{i,l}(1)) =
 \frac{\alpha_{n,1}(y_n,\check{x}_{n-1}^{i,l}(1:d),x_{n}^{i,l}(1))}{q_{n,1}(x_{n}^{i,l}(1)|\check{x}_{n-1}^{i,l}(1:d))}, \quad l\in\{1,\dots,M_d\}.
$$
The $M_d$-samples are resampled, according to the weights inclusive of the $\check{x}_{n-1}^{i,l}(1:d)$, which are denoted $\check{x}_{n-1,j}^{i,l}(1:d)$ at step $j$.
At subsequent points $j\in\{2,\dots,d\}$, one produces $M_d$-samples from $q_{n,j}$ in $\mathbb{R}$ and computes the weights, for $l\in\{1,\dots,M_d\}$
$$
G_{n,j}(\check{x}_{n-1,j-1}^{i,l}(1:d),\check{x}_{n}^{i,l}(1:j-1),x_{n}^{i,l}(j)) = \frac{\alpha_{n,j}(y_n,\check{x}_{n-1,j-1}^{i,l}(1:d),\check{x}_{n}^{i,l}(1:j-1),x_{n}^{i,l}(j))}{q_{n,j}(x_{n}^{i,l}(j)|\check{x}_{n-1,j-1}^{i,l}(1:d),\check{x}_{n}^{i,l}(1:j-1))}. 
$$
The $M_d$-samples are resampled according to the weights.
At the end of the time step, the $N$ particle systems are assigned weights
\begin{gather}
\mathbf{G}_n(\check{x}_{n-1,1:d-1}^{i,1:M_d}(1:d),\check{x}_{n}^{i,1:M_d}(1:d-1),x_{n}^{i,1:M_d}(1:d)) =\nonumber \\
\prod_{j=1}^d \Big(\frac{1}{M_d}\sum_{l=1}^{M_d}
G_{n,j}(\check{x}_{n-1,j-1}^{i,l}(1:d),\check{x}_{n}^{i,l}(1:j-1),x_{n}^{i,l}(j))\Big).\label{eq:abc}
\end{gather}
We then resample the $N$-particle systems according to the weights. The normalizing constant 
$$
\int_{\mathbb{R}^d}\Big(\prod_{k=1}^n g(x_k,y_k)f(x_{k-1},x_k)\Big)dx_{1:n}
$$ 
can be estimated by
\begin{equation}
\prod_{k=1}^n \Big(\frac{1}{N}\sum_{i=1}^N\mathbf{G}_k(\check{x}_{k-1,1:d-1}^{i,l}(1:d),\check{x}_{k}^{i,1:M_d}(1:d-1),x_{k}^{i,1:M_d}(1:d))\Big)\label{eq:ncn}.
\end{equation}
For $\varphi:\mathbb{R}^d\rightarrow\mathbb{R}$, the filter at time $n$, 
$$
\frac{\int_{\mathbb{R}^{nd}}\varphi(x_n)\prod_{k=1}^n g(x_k,y_k)f(x_{k-1},x_k)dx_{1:n}}{
\int_{\mathbb{R}^{nd}}\prod_{k=1}^n g(x_k,y_k)f(x_{k-1},x_k)dx_{1:n}}
$$ 
can be estimated by (assuming again that $\check{x}_{n}^{i,l}(1:d)$  
denote the values after resampling according to the global weights in (\ref{eq:abc}))
\begin{equation}
\frac{1}{NM_d}\sum_{l=1}^{M_{d}}\sum_{i=1}^N\varphi(\check{x}_{n}^{i,l}(1:d))\label{eq:filtn}.
\end{equation}

\subsection{Remarks}

In terms of the estimate of the filter \eqref{eq:filt1}, \eqref{eq:filtn}, we expect there to be a \emph{path degeneracy} effect for the local filters (see \cite{doucet}), especially for $d$ large, due to resampling forcing common ancestries for different
particles. 
For instance, in a worst case scenario, for a given $i\in\{1,\dots,N\}$, only one of the $M_d$ samples will be a good representation of the target filtering distribution 
at current time-step. However, one can still average over all $M_d$-samples as we have done; one can also select a single sample for estimation, if preferred. In addition, in a general setting the form of the weights $G_{n,j}$, $n\geq 2$, depends upon $\check{x}_{n-1}^{i,l}(1:d)$; there may be an additional path degeneracy effect with these samples. To an extent, this can be alleviated using dynamic resampling (e.g.~\cite{delmoral_resampling} and the references therein);
we will discuss how path degeneracy could be potentially dealt with in Section \ref{sec:path_degen} below. In addition, in some scenarios (see e.g.~\cite{poy}) the path degeneracy can betaken care of if the number of samples is quadratic in the time parameter; i.e.~$M_d=\mathcal{O}(d^2)$.

Note that we have assumed that
$$
g(x_k,y_k)f(x_{k-1},x_k) = \prod_{j=1}^d \alpha_{k,j}(y_k,x_{k-1},x_{k}(1:j)).
$$
However, this need not be the case. All one needs is a collection of functions $\alpha_{k,j}$, such that the variance (w.r.t.~the simulated algorithm) of 
\begin{equation}
\frac{g(x_k,y_k)f(x_{k-1},x_k)}{\prod_{j=1}^d \alpha_{k,j}(y_k,x_{k-1},x_{k}(1:j))}\label{eq:correct_normal}
\end{equation}
is reasonable, especially as $d$ grows. 
Then, the particles obtained at the end of the $k$-th time-step under 
$\prod_{j=1}^d \alpha_{k,j}(y_k,x_{k-1}, x_{k}(1:j))$ can 
be used as proposals with an importance sampler targeting $g(x_k,y_k)f(x_{k-1},x_k)$,
with the above ratio giving the relevant weights.
 In such a scenario, we expect the algorithm to perform reasonably well, even for large $d$; however, the construction of such functions
$\alpha_{k,j}$ may not be trivial in general.

The algorithm is easily parallelized over $N$, at least in-between global resampling times. We also note that the idea of using a particle filter within a particle filter has been used, for example, in \cite{johansen}.
The algorithm can also be thought of as a novel generalization of the island particle filter \cite{verge}. In our algorithm, one runs an entire particle filter for $d$ time steps, as the local filter, whereas, it is only
one step in \cite{verge}; as we shall see in Section \ref{sec:theory}, this appears to be critical in the high-dimensional filtering context. We also remark that, unlike the method described  in 
\cite{rebs}, the algorithm in this is article is consistent as $N$ grows.

\subsection{Dealing with Path Degeneracy}\label{sec:path_degen}

As mentioned above, the path degeneracy effect may limit the success of the proposed algorithm. 
We expect it to be of use when $d$ is maybe too large for the standard particle filter, but not overly large.
Path degeneracy can in principle be dealt with, at an increased computational cost, in the following way; in such cases one can run the algorithm simply with $N=1$.
At time 1, one may apply an Markov chain Monte Carlo (MCMC) `mutation' kernel for each local particle at each dimension step, where the invariant target density is proportional to ($j\in\{1,\dots,d\}$)
$$
\prod_{k=1}^j \alpha_{1,k}(y_1,x_{0}, x_{1}(1:k)).
$$
At subsequent time steps $n$, one uses the marginal particle filter (e.g.~\cite{poy}) and targets, up-to proportionality  for each local particle at each space-step 
$$
\sum_{l=1}^{M_d} \prod_{k=1}^j \alpha_{n,k}(y_n,\check{x}_{n-1}^{i,l}(1:d), x_{n}(1:k))
$$
also using MCMC steps with the above invariant density. Notice that the above 
expression is a Monte Carlo estimator the (unnormalised) marginal distribution of $x_{n}(1:j)$ under the model specified by the $\alpha_{n,k}$ functionals. 
Assuming an effective design of the MCMC step, 
the path degeneracy effect can be overcome, and each time-step $n$ will still has fixed (but increased) computational complexity.
The cost of this modified algorithm, assuming the cost of computing $\alpha_{n,k}$ is $\mathcal{O}(1)$ for each $n,k$,  is $\mathcal{O}(nNM_d^2d^2)$; so long as $M_d$ is polynomial in $d$, this is still a reasonable algorithm for high-dimensional problems. We note that, even though we do not analyze this
algorithm mathematically, we will implement it.

\section{Theoretical Results}\label{sec:theory}

\subsection{Consistency of Space-Time Sampler}

We will now establish that if $d,M_d\geq 1$ are fixed then   STPF will provide consistent estimates of quantities of interest of the true filter as $N$ grows. Indeed, one can prove many results
about the algorithm in this setting, such as finite-$N$ bounds and central limit theorems; however, this is not the focus of this work and the consistency result
is provided to validate the use of the algorithm. Throughout, we condition on a fixed data record and
we will suppose that 
$$
\sup_{x\in\mathbb{R}^j}|G_{1,j}(x)| <+\infty, \sup_{x\in\mathbb{R}^{d+j}}|G_{n,j}(x)| <+\infty, \quad n\geq 2.
$$
Below $\rightarrow_{\mathbb{P}}$ denotes convergence in probability as $N$ grows, where $\mathbb{P}$ denotes the law under the simulated algorithm. 
We
denote by $\mathcal{B}_b(\mathbb{R}^d)$ the class of bounded and measurable real-valued functions on $\mathbb{R}^d$.
We will write, for $n\geq 1$
$$
\pi_n(\varphi) := \frac{\int_{\mathbb{R}^{nd}}\varphi(x_n)\prod_{k=1}^n g(y_k|x_k)f(x_k|x_{k-1})dx_{1:n}}{
\int_{\mathbb{R}^{nd}}\prod_{k=1}^n g(y_k|x_k)f(x_k|x_{k-1})dx_{1:n}}
$$
and
$$
p(y_{1:n}) = \int_{\mathbb{R}^{nd}}\Big(\prod_{k=1}^n g(y_k|x_k)f(x_k|x_{k-1})\Big)dx_{1:n},
$$
so that $\pi_n$ corresponds to the filtering density of $X_n|y_{1:n}$.
The proof of the following Theorem is given in Appendix \ref{sec:prf_consis}. It ensures that the $N$ particle systems correspond to a standard particle filter
on an enlarged state space; 
once this is established standard consistency results for particle filters on general state spaces (e.g.~\cite{delmoral}) 
will complete the proof. We denote by $\rightarrow_{\mathbb{P}}$ convergence in probability.

\begin{theorem}\label{theo:consis}
Let $d,M_d\geq 1$ be fixed and let $\varphi\in\mathcal{B}_b(\mathbb{R}^d)$. Then we have for any $n\geq 2$
\begin{eqnarray*}
\frac{1}{NM_d}\sum_{l=1}^{M_{d}}\sum_{i=1}^N\varphi(\check{x}_{1}^{i,l}(1:d)) & \rightarrow_{\mathbb{P}} & 
\pi_1(\varphi),\\
\frac{1}{N}\sum_{i=1}^N\mathbf{G}_1(\check{x}_{1}^{i,1:M_d}(1:d-1),x_{1}^{i,1:M_d}(1:d)) & \rightarrow_{\mathbb{P}} & p(y_{1}),\\
\frac{1}{NM_d}\sum_{l=1}^{M_{d}}\sum_{i=1}^N\varphi(\check{x}_{n}^{i,l}(1:d)) & \rightarrow_{\mathbb{P}} & \pi_n(\varphi),\\
\prod_{k=1}^n \Big(\frac{1}{N}\sum_{i=1}^N\mathbf{G}_k(\check{x}_{k-1,1:d-1}^{i,l}(1:d),\check{x}_{k}^{i,1:M_d}(1:d-1),x_{k}^{i,1:M_d}(1:d))\Big)
 & \rightarrow_{\mathbb{P}} & p(y_{1:n}).
\end{eqnarray*}
\end{theorem}


\begin{rem}
The proof establishes that also $\frac{1}{N}\sum_{i=1}^N\varphi(\check{x}_{1}^{i,1}(1:d))$ can be used as an estimator for the filter; this may be more effective than the estimator given in the statement of the Theorem, due to the path degeneracy effect mentioned earlier.
In addition, one can assume the context described in (\ref{eq:correct_normal})
with the target not having a product structure, but the weights in (\ref{eq:correct_normal}) have controlled variance. 
 Even in this more general case one can the follow the arguments in the proof, to obtain consistency in that case (assuming the expression in \eqref{eq:correct_normal} is upper-bounded).
\end{rem}

\subsection{Stability in High-Dimensions for i.i.d.~Model}

We now come to the main objective of our theoretical analysis. We will set $N$ as fixed and consider the algorithm as $d$ grows. In order to facilitate our analysis, we will
consider approximating a probability, with density proportional to
$$
\prod_{k=1}^n\prod_{j=1}^d \alpha(x_{k}(j)).
$$
We will use the STPF with proposals $q_{n,j}(x_{n,j}|x_{n-1}(1:d),x_n(1:j)) = q(x_n(j))$.
In the case of a state-space model, this would correspond to
$$
g(x_k,y_k)f(x_{k-1},x_k) = \prod_{j=1}^d \alpha(x_{k}(j)).
$$
which would seldom occur in a real scenario. However, analysis in this context is expected to be informative for more complex scenarios as
in the work of \cite{beskos}. Note that, because of the loss of dependence on subsequent observation times, we expect that any complexity analysis with respect to $d$ to be slightly over-optimistic;
as noted the path degeneracy effect is expected to play a role in this algorithm 
in general.

We will consider the relative variance of the standard estimate of the normalizing constant  $p(y_{1:n})$, given for instance in Theorem \ref{theo:consis} which now writes as
\begin{align}
p^{N,M_d}(y_{1:n}) &= \prod_{k=1}^n \frac{1}{N}\sum_{i=1}^N \prod_{j=1}^d\frac{1}{M_d}\sum_{l=1}^{M_d}\frac{\alpha(x_k^{i,l}(j))}{q(x_k^{i,l}(j))} \nonumber \\
& \equiv  \prod_{k=1}^n \frac{1}{N}\sum_{i=1}^N \mathbf{G}_k(x_k^{i,1:M_d}(1:d))
\label{eq:nc_est_hd}.
\end{align}
The proof of the following result is given in Appendix \ref{sec:rv_res}.
Note that due to the i.i.d.~structure along time and space, all variables 
$x_k^{i,l}(j)$ can be assumed i.i.d.~from $q(\cdot)$.

\begin{prop}\label{prop:rv_res}
Assume that $$\frac{\int\alpha(x)^2/q(x)dx}{(\int\alpha(x)dx)^2}<+\infty,$$ then
$$
\mathbb{E}\Big[\Big(\frac{p^{N,M_d}(y_{1:n})}{p(y_{1:n})}-1\Big)^2\Big] =
 \Big(\frac{1}{N}\Big(\frac{1}{M_d}\frac{\int\alpha(x)^2/q(x)dx}{(\int\alpha(x)dx)^2} + \frac{M_d-1}{M_d}\Big)^d + \frac{N-1}{N}\Big)^n - 1.
$$
\end{prop}
\begin{rem}
The case $M_d=1$ corresponds, in some sense, to the standard particle filter. In this case, by Jensen's inequality, the right hand side of the above identity will diverge as $d$ grows, unless $N$ is of exponential order in $d$. As a result, we can stabilize the algorithm with an $\mathcal{O}(nd\kappa^d)$ cost, where $\kappa>1$. However, if one sets $M_d=d$, then the right hand side of the above identity will stabilize and the cost of the
algorithm is $\mathcal{O}(nNd^2)$. This provides some intuition about why our approach may be effective in high dimensions.
\end{rem}

In fact, one can say a bit more. 
We suppose that
$\alpha(x)/q(x)$ is upper and lower bounded; this typically implies that $x$ lies only on some compact subset of $\mathbb{R}$.
Denoting by $\Rightarrow$ weak convergence as $d\rightarrow\infty$ and $\mathcal{LN}(\mu,\sigma^2)$ the log-normal distribution of location $\mu$, scale $\sigma$, 
we have the following.

\begin{prop}\label{prop:conv_weights}
Let $M_d=d/c$, for some $0<c<+\infty$ and $N,n\geq 1$ fixed. Suppose that
\begin{equation}
\sigma^2 := \frac{\int\alpha(x)^2/q(x)dx}{(\int\alpha(x)dx)^2} - 1 <+\infty.\label{eq:sig_def}
\end{equation}
Then we have that $\mathbf{G}_k(X_k^{i,1:M_d}(1:d))/(\int_{\mathbb{R}}\alpha(x)dx)^{d}\Rightarrow V_{k}^i$, and subsequently 
$$
\frac{p^{N,M_d}(y_{1:n})}{p(y_{1:n})} \Rightarrow \prod_{k=1}^n \frac{1}{N}\sum_{i=1}^N V_k^i
$$
where $V_k^i\stackrel{\textrm{i.i.d.}}{\sim}\mathcal{LN}(-c\sigma^2/2,c\sigma^2)$.
\end{prop}

\begin{proof}
The result follows from \cite[Theorem 1.1]{berard} and   elementary calculations, which we omit.
\end{proof}

\begin{rem}
The result suggests that the algorithm stabilizes as $d$ grows at a $\mathcal{O}(nNd^2)$ cost. Using the continuous mapping theorem, for $N>1$, one can show that the effective sample size
(ESS) will also converge to a non-trivial random variable; see e.g.~\cite[Proof of Theorem 3.2]{beskos}. Moreover, based upon personal communication with Pierre Del Moral, we conjecture that setting $M_d=d^{1+\delta}/c$, for some $\delta>0$, the ESS converges to $N$; hence suggesting 
that $M_d=\mathcal{O}(d)$ is an optimal computational effort in this case. 
\end{rem}

\begin{rem}
An intuition behind the results is that for a standard particle filter, when run for $n$ steps with $N$ particles, the relative variance 
of the estimate for the normalizing constant grows at most linearly in  the number of steps $n$ provided $N=\mathcal{O}(n)$ (see  \cite{cerou1} for details). In the algorithm, the weights $\mathbf{G}_n$ are estimates of normalizing constants 
for the local filter, so one expects that if $M_d=\mathcal{O}(d)$,
then the algorithm should work well for large $d$. There is, however, an important point to be made. The result above assumes an i.i.d.~structure which removes any path degeneracy effect, both
within a local filter, and in the time-dependence between observations. However, in general contexts one still expects that setting $M_d$ to be a polynomial function of $d$ should allow reasonable
empirical performance. This is because the relative variance  of the normalizing constant can be controlled in such path dependent cases, with polynomial cost; see \cite{wang} for example.
\end{rem}

\begin{rem}
In the case of no global resampling, one would typically use the estimate, for $p(y_{1:n})$
$$
\frac{1}{N}\sum_{i=1}^N \prod_{k=1}^n\prod_{j=1}^d\frac{1}{M_d}\sum_{l=1}^{M_d}\frac{\alpha(x_k^{i}(j))}{q(x_k^{i}(j))}.
$$
A weak convergence result also holds in this case. 
\end{rem}

We now adopt a context of no global resampling and consider the Monte Carlo error of the following two estimates, for $n\geq 1$, $l\in\{1,\dots,M_d\}$ fixed and $\varphi\in\mathcal{C}_b(\mathbb{R})$, 
$$
\frac{1}{N}\sum_{i=1}^N \frac{\prod_{k=1}^n \mathbf{G}_k(x_k^{i,1:M_d}(1:d))}{\sum_{j=1}^N \prod_{k=1}^n \mathbf{G}_k(x_k^{j,1:M_d}(1:d))}\frac{1}{M_d}\sum_{l=1}^{M_d} \varphi(\check{x}^{i,l}_n(d))
$$
and
$$
\frac{1}{N}\sum_{i=1}^N \frac{\prod_{k=1}^n \mathbf{G}_k(x_k^{i,1:M_d}(1:d))}{\sum_{j=1}^N \prod_{k=1}^n \mathbf{G}_k(x_k^{j,1:M_d}(1:d))}  \varphi(\check{x}^{i,l}_n(d)).
$$
We remark that this is the simplest case in terms of analysis, as for example the case of when global resampling is considered  
 is seemingly more complex. 
We now give our result; the technical results for the proof can be found in Appendix
\ref{app:mc_avg}. We set $$\pi(\varphi)=\int_{\mathbb{R}}\alpha(x)\varphi(x)dx/\int_{\mathbb{R}}\alpha(x)dx.$$ 
For $\varphi\in\mathcal{B}_b(\mathbb{R})$, we denote $\|\varphi\|_{\infty}:=\sup_{x\in\mathbb{R}}|\varphi(x)|$. Also $\mathcal{C}_b(\mathbb{R})$ are the continuous and real-valued functions on $\mathbb{R}$.

\begin{theorem}
Let $M_d=d/c$, for some $0<c<+\infty$ and $n\geq 1$, $N>1$ fixed. Then we have, for any $\varphi\in\mathcal{C}_b(\mathbb{R})$, $1\leq p <+\infty$
\begin{enumerate}
\item{
$$
\lim_{d\rightarrow\infty}\mathbb{E}\Big[\Big| \sum_{i=1}^N \frac{\prod_{k=1}^n \mathbf{G}_k(X_k^{i,1:M_d}(1:d))}{\sum_{j=1}^N \prod_{k=1}^n \mathbf{G}_k(X_k^{j,1:M_d}(1:d))}\frac{1}{M_d}\sum_{l=1}^{M_d} \varphi(\check{X}^{i,l}_n(d))-\pi(\varphi)\Big|^p\Big]^{1/p} = 0
$$
}
\item{there exists an $M(p)<+\infty$, depending upon $p$ only, such that
$$
\lim_{d\rightarrow\infty}\mathbb{E}\Big[\Big| \sum_{i=1}^N \frac{\prod_{k=1}^n \mathbf{G}_k(X_k^{i,1:M_d}(1:d))}{\sum_{j=1}^N \prod_{k=1}^n \mathbf{G}_k(X_k^{j,1:M_d}(1:d))}  \varphi(\check{X}^{i,l}_n(d))-\pi(\varphi)\Big|^p\Big]^{1/p} \leq 
$$
$$
\frac{M(p)\|\varphi\|_{\infty}}{\sqrt{N}}\big[\exp\{-c\sigma^2p/2 + c\sigma^2p^2/2\}+1\big]^{1/p}
$$
where $\sigma^2$ is as in \eqref{eq:sig_def}.}
\end{enumerate}
\end{theorem}

\begin{proof}
For Case 1.~we have that by Proposition \ref{prop:conv_weights}, and the continuous mapping theorem that (after scaling the numerator and denominator by $(\int\alpha(x)dx)^d$), for each $i$
$$ 
 \frac{\prod_{k=1}^n \mathbf{G}_k(X_k^{i,1:M_d}(1:d))}{\sum_{j=1}^N \prod_{k=1}^n \mathbf{G}_k(X_k^{j,1:M_d}(1:d))} \Rightarrow
\frac{\prod_{k=1}^n V_k^i}{\sum_{j=1}^N \prod_{k=1}^n  V_k^j}.
$$
where $V_k^i\sim\mathcal{LN}(-c\sigma^2/2,c\sigma^2)$, for $\sigma^2$ as in Proposition \ref{prop:conv_weights}.
By standard importance sampling and resampling results 
(see for instance \cite{rubin})), we have that
$$
\frac{1}{M_d}\sum_{l=1}^{M_d} \varphi(\check{X}^{i,l}_n(d)) \rightarrow_{\mathbb{P}} \pi(\varphi).
$$
By Lemma \ref{lem:asymp_ind} 2., these two terms are asymptotically independent. Thus we have
$$
\sum_{i=1}^N \frac{\prod_{k=1}^n \mathbf{G}_k(X_k^{i,1:M_d}(1:d))}{\sum_{j=1}^N \prod_{k=1}^n \mathbf{G}_k(X_k^{j,1:M_d}(1:d))}\frac{1}{M_d}\sum_{l=1}^{M_d} \varphi(\check{X}^{i,l}_n(d))
\Rightarrow \pi(\varphi).
$$
The proof of 1.~is complete on noting the boundedness of the associated quantities.

For Case 2.~by Proposition \ref{prop:conv_weights}, 
the fact that $\check{X}^{i,l}_n(d)\Rightarrow V^i\sim \pi$ (see e.g.\@ \cite{rubin})
 and Lemma~\ref{lem:asymp_ind} 1.~we have
$$
\sum_{i=1}^N \frac{\prod_{k=1}^n \mathbf{G}_k(X_k^{i,1:M_d}(1:d))}{\sum_{j=1}^N \prod_{k=1}^n \mathbf{G}_k(X_k^{j,1:M_d}(1:d))}  \varphi(\check{X}^{i,l}_n(d))
\Rightarrow \sum_{i=1}^N\frac{\prod_{k=1}^n V_k^i}{\sum_{j=1}^N \prod_{k=1}^n  V_k^j}\varphi(V^i)
$$
where the $V^i$ are independent of the $V_k^i$ and have a distribution that has density $\pi$. Then, by the   boundedness of the associated quantities we have
$$
\lim_{d\rightarrow\infty}\mathbb{E}\Big[\Big| \sum_{i=1}^N \frac{\prod_{k=1}^n \mathbf{G}_k(X_k^{i,1:M_d}(1:d))}{\sum_{j=1}^N \prod_{k=1}^n \mathbf{G}_k(X_k^{j,1:M_d}(1:d))}  \varphi(\check{X}^{i,l}_n(d))-\pi(\varphi)\Big|^p\Big]^{1/p}
$$
$$
= \mathbb{E}\Big[\Big| \sum_{i=1}^N\frac{\prod_{k=1}^n V_k^i}{\sum_{j=1}^N \prod_{k=1}^n  V_k^j}\varphi(V^i)-\pi(\varphi)\Big|^p\Big]^{1/p}.
$$
The proof can now be completed by the same calculations as in the proof of \cite[Theorem 3.3]{beskos} and are hence omitted.
\end{proof}

\begin{rem}
The main points are, first, that the error in estimation of fixed-dimensional marginals is independent of $d$ and, second, that averaging over the local particle cloud seems to help in high dimensions. 
We repeat that the scaling for $M_d$ that stabilises the weights for the 
global filter 
may be over-optimistic for more general models, due to the loss of a path-degeneracy effect over the observation times in the i.i.d.~case. 
\end{rem}

\subsection{Stability in High Dimensions for Markov Model}

We now consider a more realistic scenario for our analysis in high-dimensions. In order to read this Section, one will need to consult Appendices \ref{sec:prf_consis} and \ref{app:prf_nc}; this Section can be skipped with no loss in continuity.

We consider the interaction of the dimension and the time parameter in the behaviour
of the algorithm. We will now list some assumptions and notations needed to describe the result.
\begin{hypA}\label{hyp:pstruc}
For every $n\geq 1$ we have
$$
g(x_n,y_n)f(x_{n-1},x_n) = \prod_{j=1}^d h(y_n,x_{n}(j))k(x_{n}(j-1),x_n(j))
$$
where $h:\mathbb{R}^k\rightarrow\mathbb{R}^+$, $x_n(0) = x_{n-1}(d)$ and for every $x\in\mathbb{R}$,
$\int_{\mathbb{R}} k(x,x')dx' = 1$.
\end{hypA}
It is noted that even under (A\ref{hyp:pstruc}) a standard particle filter 
which propagates all $d$ co-ordinates together 
 may degenerate as $d$ grows. However, as we will remark, the STPF can stabilize
under assumptions, even if $N=1$.
Our algorithm will use the Markov kernels $k(x_{n}(j-1),x_n(j))$ as the proposals.
Define the semigroup, for $p\geq 1$:
$$
\hat{q}_p(x_{p-1},dx_p) = f(x_{p-1},x_p)g_p(x_p)dx_p
$$
where $g_p(x_p)=g(y_p,x_p)$. For $\varphi\in\mathcal{B}_b(\mathbb{R}^d)$ define
\begin{equation}
\hat{q}_{p,n}(\varphi)(x_{p}) =  \int \hat{q}_{p+1}(x_{p},dx_{p+1})\times\cdots\times\hat{q}_{n}(x_{n-1},dx_{n}) \varphi(x_n)\label{eq:semi1_def}.
\end{equation}

\begin{hypA}\label{hyp:semigroup}
There exists a $c<\infty$, such that for every $1\leq p<n$ and $d\geq 1$
$$
\sup_{x,y}\frac{\hat{q}_{p,n}(1)(x)}{\hat{q}_{p,n}(1)(y)} \leq c.
$$
\end{hypA}
Note (A\ref{hyp:semigroup}) is fairly standard in the literature (e.g.~\cite{delmoral1}) and given (A\ref{hyp:pstruc}) it will hold under some simple assumptions on $h$ and $k$.

Now, we will consider the global filter with $N$ particles,  
as standard results in the literature can provide immediately CLTs and SLLNs
for quantities of interest. We will then investigate the effect of the dimensionality 
$d$ on the involved terms.
Consider the standard estimate for the normalising constant for the global filter
$$
\boldsymbol{\gamma}_n^N(1) := \prod_{p=1}^{n-1} \boldsymbol{\eta}_p^N(\mathbf{G}_p) 
$$
when $\boldsymbol{\eta}_p^N(\cdot)$ simply denotes Monte-Carlo 
averages over the $N$ particle systems at time $p$,
see Appendix \ref{sec:prf_consis} for analytic definitions.
From standard particle filtering theory, we have that 
$\boldsymbol{\eta}_p^N(\cdot)$ is an unbiased estimator of the corresponding
limiting quantity, denoted $\boldsymbol{\gamma}_n(1)$, see e.g.\@ \cite[Theorem 7.4.2]{delmoral}. Also, under our assumptions, one has the following  CLT as $N\rightarrow \infty$
(see \cite[Proposition 9.4.2]{delmoral})
\begin{equation}
\label{eq:CLT}
\sqrt{N}\Big(\frac{\boldsymbol{\gamma}_n^N(1)}{\boldsymbol{\gamma}_n(1)}-1\Big) \Rightarrow \mathcal{N}(0,\sigma^2_n)
\end{equation}
where $\mathcal{N}(0,\sigma^2)$ is the one dimensional
normal distribution with zero mean and variance $\sigma^2$,
and 
$$
\sigma^2_n = \frac{1}{\boldsymbol{\gamma}_n(1)^2}\sum_{p=1}^n \boldsymbol{\gamma}_p(1)^2
\boldsymbol{\eta}_p\bigg(\Big(\mathbf{Q}_{p,n}(1)-\boldsymbol{\eta}_p(\mathbf{Q}_{p,n}(1))\Big)^2\bigg).
$$
All bold terms correspond to standard Feynman-Kac quantities and are defined in Appendix \ref{sec:prf_consis}.
We also show in Appendix \ref{sec:prf_consis} that the normalising 
constant of the global filter coincides with the one of the original filter of interest, 
that is
$$
 \boldsymbol{\gamma}_n(1) \equiv \gamma_n(1) = \int \prod_{p=1}^{n-1}g_p(x_p) f(x_{p-1},x_p)dx_{1:p} =  p(y_{1:n-1})
$$
Thus, (\ref{eq:CLT}) provides in fact a CLT for the 
estimate of STPF for $p(y_{1:n-1})$ proposed in Theorem \ref{theo:consis}.
%
%
%

We have the following result, whose proof is in Appendix \ref{app:prf_nc}:

\begin{theorem}\label{theo:nc_is_ok}
Assume (A\ref{hyp:pstruc}-\ref{hyp:semigroup}). Then there exist a $\bar{c}<\infty$ such that for any $n,d\geq 1$ and any $M_d\geq \bar{c}d$
$$
\sigma^2_n \leq n\bar{c}\Big(\frac{d}{M_d} + 1\Big).
$$
\end{theorem}

\begin{rem}
Our result establishes that the asymptotic in $N$ variance of the relative value of the normalizing constant estimate grows at most linearly in $n$ and, if $M_d=\mathcal{O}(d)$
does not grow with the dimension. The cost of the algorithm is $\mathcal{O}(nNd^2)$. 
The linear growth in time  is a standard result in the literature (see \cite{delmoral1}) and one does not expect to do better than this.
Note, that a particular model structure is chosen and one expects a higher cost in more general problems.
\end{rem}

\begin{rem}
We expect that to show that the error in estimation of the filter is time uniform, under (A\ref{hyp:pstruc}), that one will need to set $M_d=\mathcal{O}(d^2)$.
This is because one is performing estimation on the path of the algorithm; see \cite[Theorem 15.2.1 and Corollary 15.2.2]{delmoral1}.
Indeed, one can be even more specific; if $N=1$, then one can show that, under  (A\ref{hyp:pstruc}-\ref{hyp:semigroup}) that the $L_p$-error
associated to the estimate of the filter (applied to a bounded test function in $\mathbb{R}^d$) at time $n$ is upper-bounded by $c\|\varphi\|_{\infty}d/\sqrt{M_d}$ (via
 \cite[Theorem 15.2.1, Corollary 15.2.2]{delmoral1}) with $c$ independent of $d$ and $n$. Thus setting $M_d=\mathcal{O}(d^2)$, the upper-bound depends on $d$
only through $\|\varphi\|_{\infty}$.
\end{rem}

\section{Numerical Results}\label{sec:numerics}

\subsection{Example 1}
\label{sec:exam1}

We consider the following simple model. Let $X_n\in\mathbb{R}^d$ be such that
we have 
$X_0 = \mathbf{0}_d$ (the $d$-dimensional vector of zeros) and
\begin{equation*}
  X_n(j) =
  \sum_{i=1}^{j-1}\beta_{d-j+i+1} X_n(i) +
  \sum_{i=j}^d \beta_{i-j+1}X_{n-1}(i) + \epsilon_n
\end{equation*}
where $\epsilon_n\stackrel{\textrm{i.i.d.}}{\sim}\mathcal{N}(0, \sigma_x^2)$
and $\beta_{1:d}$ are some known static parameters. For the observations, we
set
\begin{equation*}
  Y_n = X_n + \xi_n
\end{equation*}
where $\xi_n(j)\stackrel{\textrm{i.i.d.}}{\sim}\mathcal{N}(0, \sigma_x^2)$,
$j\in\{1,\dots,d\}$. It is easily shown that this linear Gaussian model has
the structure \eqref{eq:hmm_struc}.

We consider the standard particle filter and the STPF. The
data are simulated from the model with 
 $\sigma^2_x = \sigma^2_y = 1$ and $n = 1000$
$d$-dimensional observations. These parameters are also used within the filters. Both filters use
the model  transitions as the proposal and the likelihood function as the potential.
For STPF we use $N = 1000$ and
$M_d = 100$, and for 
the  particle filter algorithm we use $NM_d$ particles. Adaptive resampling is used in all situations (with appropriate
adjustment to the formula of calculating the weights for each of the $N$
particles, as well as the estimates). Some results for $d\in\{10,100,1000\}$
are presented in Figures~\ref{fig:exam1 mean} to~\ref{fig:exam1 var}.

\begin{figure}[!ht]
  \centering
  \includegraphics[width=\linewidth]{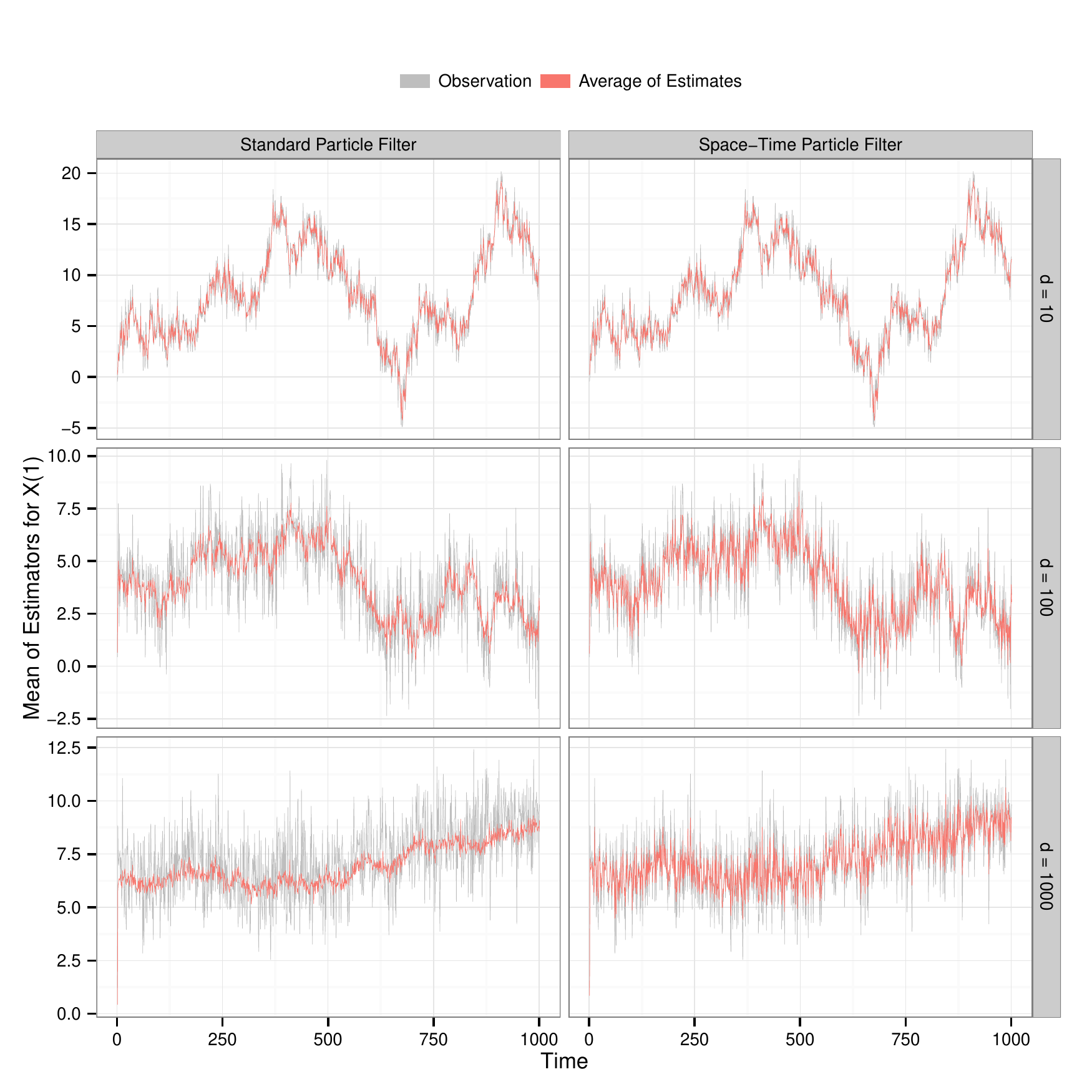}
  \caption{Mean of estimators of $X_n(1)$ for Example~\ref{sec:exam1} across
    100 runs.}
  \label{fig:exam1 mean}
\end{figure}

\begin{figure}[!ht]
  \centering
  \includegraphics[width=\linewidth]{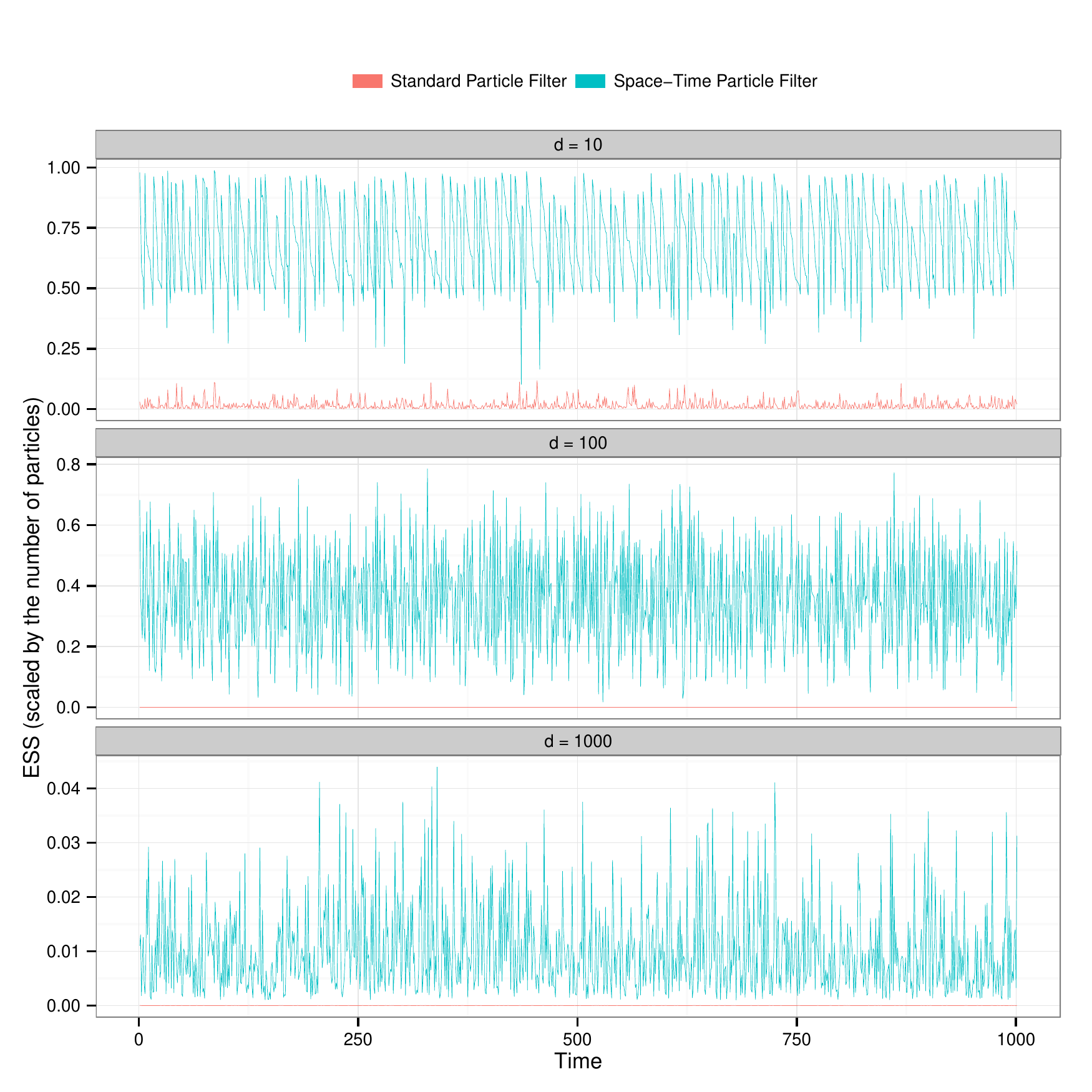}
  \caption{Effective Sample Size plots for Example~\ref{sec:exam1} from a
    single run.}
  \label{fig:exam1 ess}
\end{figure}

\begin{figure}[!ht]
  \centering
  \includegraphics[width=\linewidth]{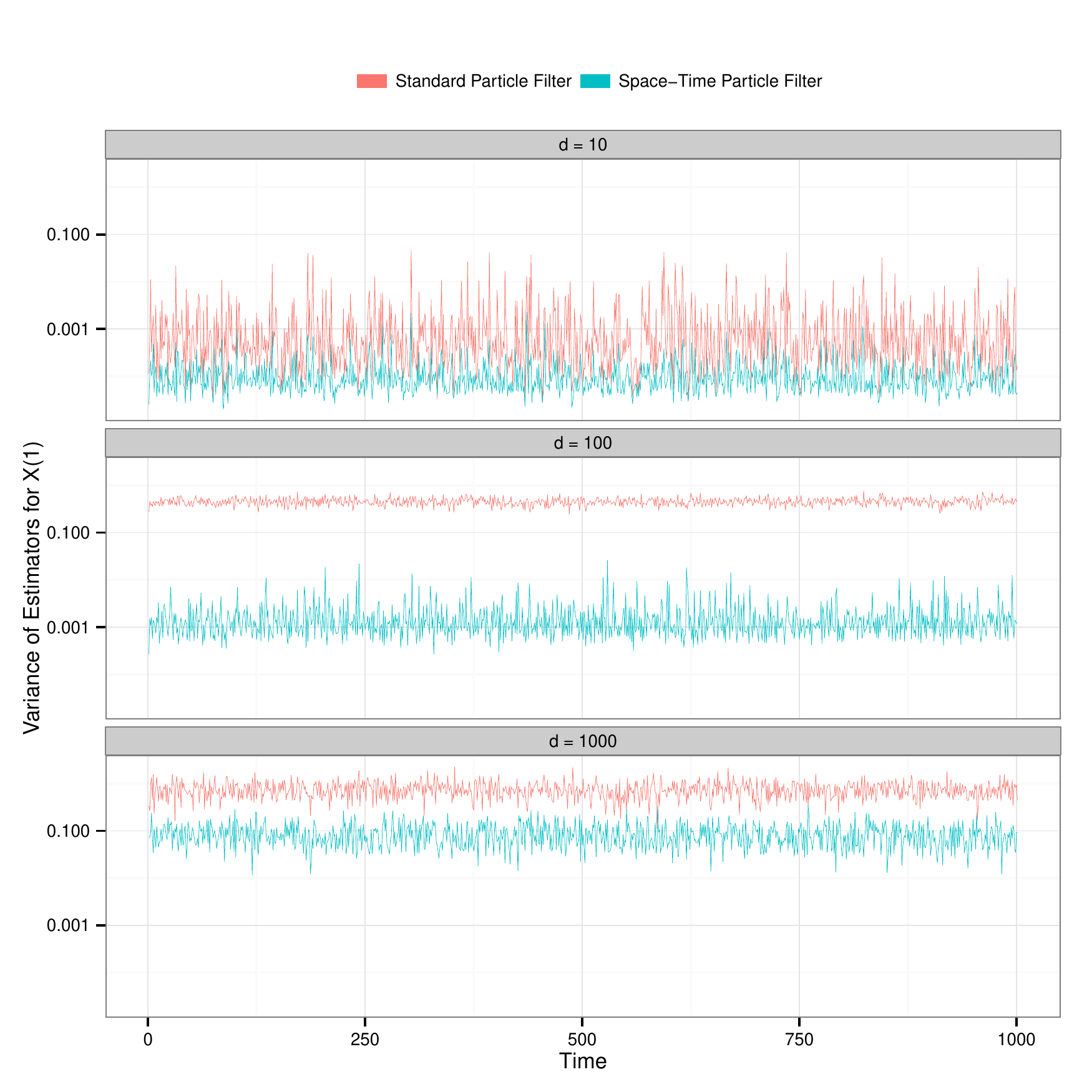}
  \caption{Variance (on logarithm scale) for estimators of $X_n(1)$ for
    Example~\ref{sec:exam1} across 100 runs.}
  \label{fig:exam1 var}
\end{figure}

The averages of estimators per time step (for the posterior mean of the first co-ordinate $X_n(1)$ given all date up to time $n$) across
100 separate algorithmic runs are illustrated in Figure~\ref{fig:exam1 mean}.
For STPF, the estimator corresponds to the double average over $M_d$, $N$ as shown in Section  \ref{sec:algo}. 
The figure shows that the particle filter collapses when the dimension become
moderate or large. It is unable to provide meaningful estimates when $d =
1000$ (as the estimates completely lose track of the observations). In contrast, the STPF performs reasonably well in all three cases.
In Figure~\ref{fig:exam1 ess} we can observe the ESS (scaled by the number of
particles) for each time step of the two algorithms. The standard filter
struggles significantly even in the case $d = 10$ and it collapses when $d =
1000$. The performance of the new algorithm is deteriorating (but not
collapsing) when the dimension increases. This is inevitably due to the path
degeneracy effect that we have mentioned. These conclusions are further
supported in Figure~\ref{fig:exam1 var} where the variance per time step
for the estimators of the posterior mean of the first co-ordinate $X_n(1)$ 
(given the data up to time $n$)
across 100 runs is displayed.

\subsection{Example 2}
\label{sec:exam2}

\subsubsection{Model and Simulation Settings}

We consider the following model on a two-dimensional graph, which follows that
described in \cite{rebs}. Let the components of state $X_n$ be indexed by
vertices $v \in V$, where $V = \{1,\dots,d\}^2$. The dimension of the model is
thus $d^2$. The distance between two vertices, $v = (a, b)$ and $u = (c, d)$,
is calculated in the usual Euclidean sense, $D(v, u) = \sqrt{(a - c)^2 + (b -
d)^2}$. At time $n$, $X_n(v)$ follows a mixture distribution,
\begin{equation*}
  f(x_{n-1},x_n(v)) = \sum_{u\in N(v)} w_u(v) f_u(x_{n-1}(u), x_n(v))
\end{equation*}
where $N(v) = \{u : D(v, u) \le r\}$ for $r \ge 1$ is the neighborhood of
vertex $v$. For observations,
\begin{equation*}
  Y_n = X_n + \xi_n
\end{equation*}
where $\xi_n(v)$, $v \in V$ are i.i.d.\ $t$-distributed random variables with
degree of freedom $\nu$.

In this example, we use a Gaussian mixture with component mean $X_{n-1}(u)$
and unity variance. The weights are set to be $w_u(v) \propto 1 / (D(v, u) +
\delta)$ and $\sum_{u\in N(v)}w_u(v) = 1$. In other words, when $\delta\to 0$,
each vertex evolves as a Gaussian random walk itself. We simulated data from
model $r = 1$, $\delta = 1$, $\nu = 10$ and $d = 32$. It results in a 1024
dimensional model. These parameters are also used in the filters.

We will compare the standard particle filter, the STPF, the marginal STPF
algorithm (as described in Section~\ref{sec:path_degen}) and the block particle
filter (BPF) in \cite{rebs} (notice that the block particle filter is characterised by space varying bias, by construction). The simulations for the STPF versions are done with $N = M_d = 100$.  The
number of particles for the standard particle filter and BPF are $NM_d$.  For
the marginal algorithm, we also simulated with $N = 1$ and $M_d = 1000$. The
block size of BPF is set to be $b^2$, $b \in
\{1,\dots,d\}$, and it is partitioned such that each block is itself a square.
The MCMC moves of the marginal algorithm are simple Gaussian random walks with
standard deviation (the scale) being $0.5$. The optimal block size in
\cite{rebs} is about $b = 7$ for ten thousand particles and a two-dimensional
graph. Thus, we considered the cases $b = 4$ and~$8$, the two nearest integers such
that $d$ is divisible by $b$.

\subsubsection{Results}

A single run takes around 2 minutes for the standard particle filter and the
block filter on an Intel Xeon W3550 CPU, with four cores and eight threads. It
takes around 10 minutes for the STPF. It takes about 40 minutes for
the marginal algorithm with $N = 1$ and $M_d = 1000$, and about 7 hours for $N
= M_d = 100$.

The standard particle filter performs poorly and cannot provide adequate
estimates (similar to the $d = 1000$ case in the previous example). In
Figure~\ref{fig:exam2 var}, we observe the variance per time step of the
estimators for two vertices, across 30 runs. The first vertex, $X_n(3, 3)$ is
not on the boundary of either block size and the second, $X_n(8, 8)$ is on the
boundary of both block sizes. In either case, the STPF significantly
outperforms the block filter, albeit under slightly longer run times. The STPF does not collapse in high-dimensions, but perhaps does not have
excellent performance. The marginal STPF performs very well, but the
computational time is substantially higher than all of the other algorithms.
However, with $N = 1$ and $M_d=\mathcal{O}(d)$, the marginal STPF provides a
good balance between performance and computational cost in challenging
situations where the path degeneracy may hinder successful application of the new
algorithm.

The block filter variance for $X_n(8, 8)$ (boundary vertex) is about
  twice that of $X_n(3, 3)$ while the new algorithm performs equally well for
  both cases. 


\begin{figure}[!ht]
  \centering
  \includegraphics[width=\linewidth]{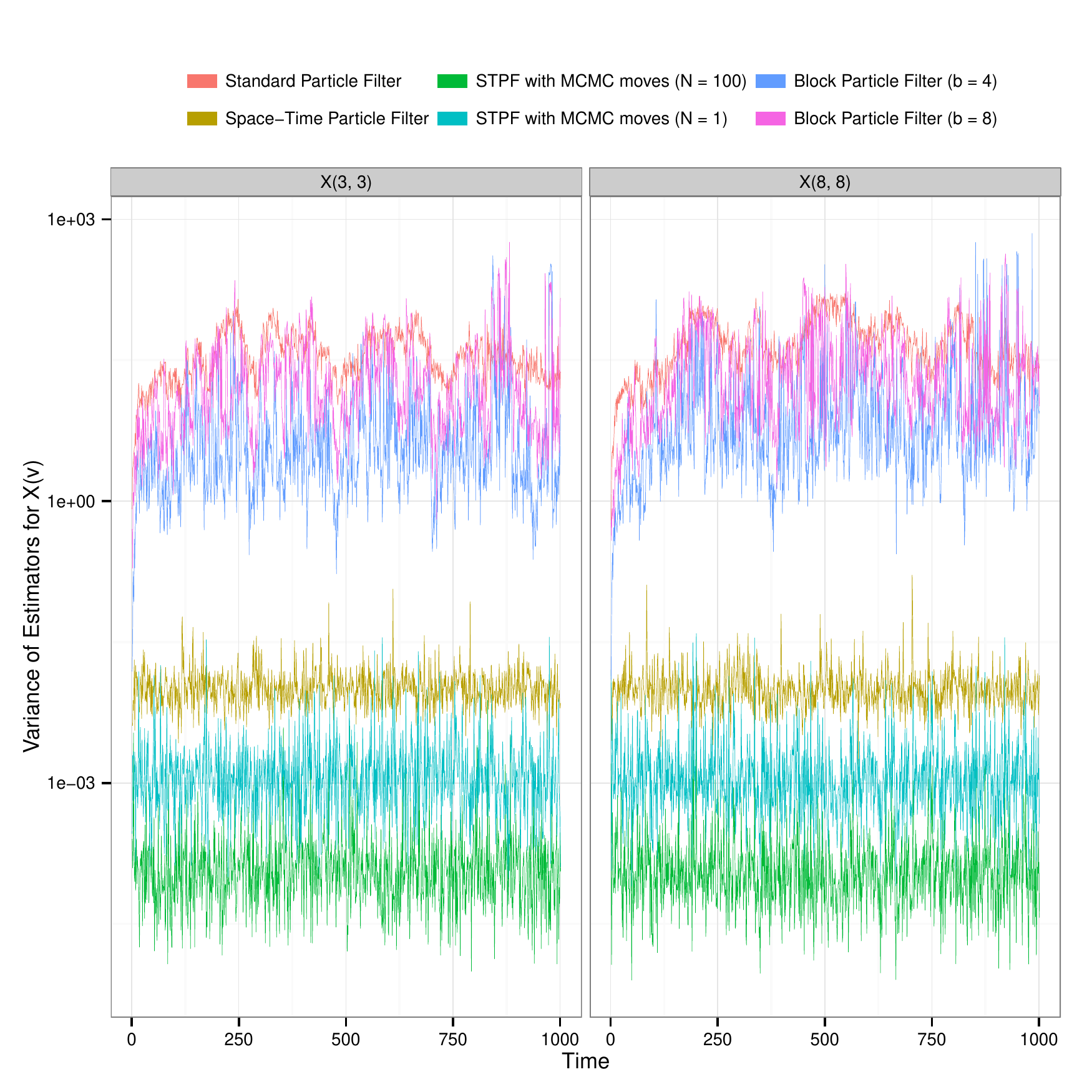}
  \caption{Variance plots (on logarithm scale) for estimators of $X_n(3, 3)$
    and $X_n(8, 8)$ for Example~\ref{sec:exam2}. The variances are estimated
    from 100 simulations for each algorithm.}
  \label{fig:exam2 var}
\end{figure}

\section{Summary}\label{sec:summary}

In this article we have considered a novel class of particle algorithms for high-dimensional filtering problems and investigated both theoretical and practical aspects of the algorithm. 
We believe the article opens new directions in an important and challenging 
Monte-Carlo problem, and several aspects of the method remain to be investigated 
in future research.
There are indeed several possible extensions to the work in this article. In particular, an analysis of the algorithm when the structure of the state-space model is more complex than the structures considered in this
article. We expect that in such scenarios, that the cost of the algorithm should increase, but only by a polynomial factor in $d$. In addition, the interaction of dimension and time behaviour is of particular interest.

\subsubsection*{Acknowledgements}

Ajay Jasra and Yan Zhou were supported by ACRF tier 2 grant R-155-000-143-112. We thank Pierre Del Moral for many useful conversations on this work.

\appendix

\section{Proof of Proposition \ref{prop:rv_res}}\label{sec:rv_res}

\begin{proof}[Proof of Proposition \ref{prop:rv_res}]
We set
\begin{equation*}
X= 
\frac{1}{M_d}\sum_{l=1}^{M_d}\frac{\alpha(x_1^{i,l}(1))}{q(x_1^{1,l}(1))}\big/\int\alpha(x)dx,\quad 
I = 
\frac{1}{N}\sum_{i=1}^N \prod_{j=1}^d\frac{1}{M_d}\sum_{l=1}^{M_d}\frac{\alpha(x_1^{i,l}(j))}{q(x_1^{i,l}(j))}\big/\int\alpha(x)dx.
\end{equation*}
Notice that $\mathbb{E}[I]=\mathbb{E}[X]=1$,  so that due to the i.i.d.\@ structure along $j$ we have that
\begin{align*}
\mathbb{E}[I^2] = \frac{1}{N}\big( \mathbb{E}[X^2]  \big)^d
+ \frac{N-1}{N}
\end{align*}
Also, due to the i.i.d.\@ structure along $j,l$ we have
\begin{align*}
\mathbb{E}[X^2]=  \frac{1}{M_d}\frac{\int a^2(x)/q(x)dx}{\big( \int a(x)dx \big)^2}
+ \frac{M_d-1}{M_d}.
\end{align*}
Finally, we have that, due to i.i.d. structure along $n$,
\begin{align*}
\mathbb{E}\Big[\Big(\frac{p^{N,M_d}(y_{1:n})}{p(y_{1:n})}-1\Big)^2\Big]
 &= \mathbb{E}\Big[\Big(\frac{p^{N,M_d}(y_{1:n})}{\big(\int\alpha(x)dx\big)^{nd}}\Big)^2\Big] - 1 \\ 
 &= ( \mathbb{E}[I^2])^n - 1.
\end{align*}
A synthesis of the above three equations gives the required result.
\end{proof}

\section{Proof of Theorem \ref{theo:consis}}\label{sec:prf_consis}

\subsection{Further Notation}
In order to prove Theorem \ref{theo:consis}, we will first introduce another round of notations. Let $(E_n,\mathscr{E}_n)_{n\geq 0}$ be a sequence of measurable spaces endowed with a countably generated $\sigma$-field $\mathscr{E}_n$. The set $\mathcal{B}_b(E_n)$ denotes the class of bounded $\mathscr{E}_n/\mathbb{B}(\bbR)$-measurable functions on $E_n$ where $\mathbb{B}(\mathbb{R})$ is the Borel $\sigma$-algebra on $\mathbb{R}$.
 We will consider non-negative operators $K : E_{n-1} \times \mathscr{E}_n \rightarrow \bbR_+$ such that for each $x \in E_{n-1}$ the mapping $A \mapsto K(x, A)$ is a finite non-negative measure on $\mathscr{E}_n$ and for each $A \in \mathscr{E}_n$ the function $x \mapsto K(x, A)$ is $\mathscr{E}_{n-1} / \mathbb{B}(\mathbb{R})$-measurable; the kernel $K$ is Markovian if $K(x, dy)$ is a probability measure for every $x \in E_{n-1}$.
For a finite measure $\mu$ on $(E_{n-1},\mathscr{E}_{n-1})$ and Borel test function $f \in \mathcal{B}_b(E_n)$ we define
\begin{equation*}
    \mu K  : A \mapsto \int K(x, A) \mu(dx);\quad 
    K f :  x \mapsto \int f(y) \, K(x, dy).
\end{equation*}

\subsection{Feynman-Kac Model on Enlarged Space}
We will define a Feynman-Kac model on an appropriate enlarged space.
That is, one Markov transition on the enlarged space will correspond to one 
observation time and will collect all $d$ space-steps of the local filter for this 
time-step. Some care is needed with the notation, as we need to keep track 
of the development of the co-ordinates at time $n$, together with the states at 
time $n-1$ as the latter are involved in the proposal.

\textbf{Time-Step 1}: Consider observation time 1 of the algorithm.  We define
a sequence of random variables $Z_{1,j}^{l}$ with $j\in\{1,\dots,d+1\}$, $1\le l \le M_d$, such that  
$Z_{1,j}^{l}\in\mathbb{R}^j$, for $j\in\{1,\dots,d\}$,  and  $Z_{1,d+1}^{l}\in\mathbb{R}^d$. For $j\in\{1,\dots,d\}$
we will write the co-ordinates of $Z_{1,j}^{l}$ as $(Z_{1,j}^{l}(1),\dots,Z_{1,j}^{l}(j))$, with the obvious extension for the  case $j=d+1$. 
As $x_0$ is fixed, we will drop it from our notations, as will become clear below.
Also, for simplicity we simply write $q(\cdot)$ instead of the analytical 
$q_{1,j}(\cdot)$ as the subscripts are implied by those of $Z_{1,j}$.
We follow this convention throughout Appendix \ref{sec:prf_consis}.
We define the following sequence of
Markov kernels corresponding to the proposal for the co-ordinates at the first time step:
\begin{align*}
M_{1,1}(dz_{1,1}) &= q(z_{1,1})dz_{1,1}, \quad j=1,\\
M_{1,j}(z_{1,j-1},dz_{1,j}) &= q(z_{1,j}(j)|z_{1,j-1})dz_{1,j}(j)\,\delta_{\{z_{1,j-1}\}}(dz_{1,j}(1:j-1)), \quad j\in\{1,\dots,d\}, \\
M_{1,j}(z_{1,j-1},dz_{1,j}) &= \delta_{\{z_{1,j-1}\}}(dz_{1,j}), \quad j=d+1.
\end{align*}
Next, we will take under consideration the weights and the resampling.
For $j\in\{1,\dots,d\}$ and a probability measure $\mu$ on $\mathbb{R}^j$ define
\begin{equation*}
\Phi_{1,j+1}(\mu)(dz) = \frac{\int_{\mathbb{R}^j} \mu(dz') G_{1,j}(z') M_{1,j+1}(z',dz)}{\int_{\mathbb{R}^j} \mu(dz') G_{1,j}(z')}.
\end{equation*}
For the local particle filter in observation time 1, write the un-weighted empirical measure 
$$
\eta_{1,j}^{M_d}(dz) = \frac{1}{M_d}\sum_{l=1}^{M_d}\delta_{z_{1,j}^{l}}(dz), 
\quad j\in\{1,\dots,d\}.  $$
We also consider all random variables involved at time-step 1 and set $$\mathbf{z}_1=(z_{1,1}^{1:M_d},\dots,z_{1,d+1}^{1:M_d}).$$
The joint law of the samples required by the local filter is
\begin{equation}
\boldsymbol{\eta}_1(d\mathbf{z}_1) = 
\Big(\prod_{l=1}^{M_d}M_{1,1}(dz_{1,1}^{l})\Big)
\Big(\prod_{j=2}^{d+1}\prod_{l=1}^{M_d}\Phi_{1,j}(\eta_{1,j-1}^{M_d})(dz_{1,j}^{l})\Big)\Big).\label{eq:eta_1_def}
\end{equation}
Notice, that in the notation we have established herein, 
the potential $\mathbf{G}_1$ defined in the main text 
can now equivalently be expressed as
%
\begin{equation}
\mathbf{G}_1(\mathbf{z}_1) = \prod_{j=1}^d \eta_{1,j}^{M_d}(G_{1,j})\label{eq:g_1_def}.
\end{equation}
We also set $z_{1,d+1}^{l}(1)=z_{1,d+1}^{l}$.

\textbf{Time-Step $n\ge 2$}: At subsequent observation times, $n\geq 2$, we again 
work with variables denoted  $Z_{n,j}^{l}$, with $j\in\{1,\dots,d+1\}$,
but this time we have to keep track of the corresponding paths at time $n-1$,
thus we will use the notation $Z_{n,j}^{l}=(Z_{n,j}^{l,+},Z_{n,j}^{l,-})$, 
with $Z_{n,j}^{l,+}\in \mathbb{R}^j$, $Z_{n,j}^{l,-}\in\mathbb{R}^{d}$, 
with the latter component referring to the `tail' at time $n-1$ of the path found
at $Z_{n,j}^{+}$ at time $n$ and space position $j$. So, we have
$Z_{n,j}^{l}\in\mathbb{R}^{j+d}$, $j\in\{1,\dots,d\}$ and 
$Z_{n,d+1}^{l}\in\mathbb{R}^{2d}$. 
We define the following sequence of
kernels:
\begin{align*}
M_{n,1}(z_{n-1,d+1}^{+},dz_{n,1}) &= q(z_{n,1}^{+}|z_{n-1,d+1}^{+})dz_{n,1}^{+}\,\delta_{\{z_{n-1,d+1}^{+}\}}(dz_{n,1}^{-}), \quad j=1,
\\
M_{n,j}(z_{n,j-1},dz_{n,j}) &= q(z_{n,j}^{+}(j)|z_{n,j-1})dz_{n,j}^{+}(j)\, \delta_{\{z_{n,j-1}^{+}\}}(dz_{n,j}^{+}(1:j-1)) \\
&\qquad \qquad \qquad \cdot \delta_{\{z_{n,j-1}^{-}\}}(dz_{n,j}^{-}),\quad j\in\{1,\dots,d\},
\\
M_{n,d+1}(z_{n,d},dz_{n,d+1}) &= \delta_{\{z_{n,d}\}}(dz_{n,d+1}),\quad 
j=d+1.
\end{align*}
For $j\in\{2,\dots,d\}$ and a probability measure $\mu$ on $\mathbb{R}^{j+d}$ define the measure on $\mathbb{R}^{\min\{j+1,d\}+d}$
$$
\Phi_{n,j+1}(\mu)(dz) = \frac{\int \mu(dz') G_{n,j}(z') M_{n,j+1}(z',dz)}{\int \mu(dz') G_{n,j}(z')}.
$$
For the local particle filter at space-step $j$, we write the  empirical measure 
$$
\eta_{n,j}^{M_d}(dz) = \frac{1}{M_d}\sum_{l=1}^{M_d}\delta_{z_{n,j}^{l}}(dz), 
\quad j\in\{1,\dots,d\}.
$$
Set $\mathbf{z}_n=(z_{n,1}^{1:M_d},\dots,z_{n,d+1}^{1:M_d})$.
The transition law of all involved samples in the local particle filter is
\begin{equation}
\mathbf{M}_n(\mathbf{z}_{n-1}, d\mathbf{z}_n) = 
\Big(\prod_{l=1}^{M_d}M_{n,1}(z_{n-1,d+1}^{l,+},dz_{n,1}^{l})\Big)
\Big(\prod_{j=2}^{d+1}\prod_{l=1}^{M_d}\Phi_{n,j}(\eta_{n,j-1}^{M_d})(dz_{n,j}^{l})\Big)\Big).\label{eq:m_n_def}
\end{equation}
Then, we will work with the potential
\begin{equation}
\mathbf{G}_n(\mathbf{z}_n) = \prod_{j=1}^d\eta_{n,j}^{M_d}(G_{n,j}).\label{eq:g_n_def}
\end{equation}

The algorithm described in Section \ref{sec:algo} corresponds to 
a standard particle filter approximation (with $N$ particles) of a Feynman-Kac model specified 
by the initial distribution (\ref{eq:eta_1_def}), the Markovian transitions
(\ref{eq:m_n_def}) and the potentials in (\ref{eq:g_1_def}), (\ref{eq:g_n_def}).
Thus, for the Monte-Carlo algorithm with $N$ particles,
set $\boldsymbol{\eta}_n^N$ for the $N$-empirical measure of $\mathbf{z}_n^{1:N}$ and set, for $\mu$ a probability measure, $n\geq 2$
$$
\boldsymbol{\Phi}_n(\mu)(d\mathbf{z}) = \frac{\int \mu(d\mathbf{z}')\mathbf{G}_{n-1}(\mathbf{z}')\mathbf{M}_n(\mathbf{z}',d\mathbf{z})}{\int \mu(d\mathbf{z}')\mathbf{G}_{n-1}(\mathbf{z}')}.
$$
Then our global filter samples from the path measure, up-to observation time $n$
$$
\Big(\prod_{i=1}^N\boldsymbol{\eta}_1(d\mathbf{z}_1^i)\Big)\Big(\prod_{k=2}^n\prod_{i=1}^N \boldsymbol{\Phi}_k(\boldsymbol{\eta}_{k-1}^N)(d\mathbf{z}_k^i)\Big)
$$
not including resampling at observation time $n$.
We  use the standard definition of the normalising constant  for any $n\geq 1$
\begin{equation}
\boldsymbol{\gamma}_n(\varphi) = \int  \boldsymbol{\eta}_1(d\mathbf{z}_1) \prod_{p=2}^n \mathbf{G}_{p-1}(\mathbf{z}_{p-1}) \mathbf{M}_p(\mathbf{z}_{p-1},d\mathbf{z}_p)\varphi(\mathbf{z}_n)\label{eq:island_nc}
\end{equation}
and set
\begin{equation}
\boldsymbol{\eta}_n(\varphi) = \frac{\boldsymbol{\gamma}_n(\varphi)}{\boldsymbol{\gamma}_n(1)}\label{eq:island_pred},
\end{equation}
thus $\boldsymbol{\eta}_n$ corresponds to the predictive distribution at time $n$
for the global filter. 
Notice, that from (\ref{eq:island_nc}), we can equivalently write for the unnormalised 
measure
\begin{align}
\label{eq:ante}
\boldsymbol{\gamma}_n(\boldsymbol{\varphi}) & = \boldsymbol{\eta}_1(\mathbf{G}_1 
 \mathbf{M}_2( \mathbf{G}_2\mathbf{M}_3 \cdots (\mathbf{G}_{n-1}
 \textbf{M}_n(
 \boldsymbol{\varphi})))).
\end{align} 

\subsection{Calculation of Quantities for Global Filter}
\label{sub:calc}
We consider functions of the particular form
\begin{equation*}
\boldsymbol{\phi}(\mathbf{z}_p) = \frac{1}{M_d}\sum_{l=1}^{M_d}\phi(z_{p,d+1}^{l,+}) , \quad \phi\in\mathcal{B}_b(\mathbb{R}^d).
\end{equation*}
For functions of the above type, we write $\boldsymbol{\phi}\in \mathcal{A}_p$.
We will illustrate that upon application on this family,  
several Feynman-Kac quantities of the global model (with 
signal dynamics $\boldsymbol{\eta}_1,\mathbf{M}_2$,$\ldots,$ and potentials
$\mathbf{G}_1,\mathbf{G}_2\dots$) coincide with those of 
the original model of interest (with signal dynamics $f_1,f_2,\ldots$ and potentials
$g_1,g_2,\ldots$). 
In particular we calculate $\mathbf{M}_p(\mathbf{G}_p \boldsymbol\phi)$
as, from (\ref{eq:ante}), it is the building block for other expressions.
Notice we can write
\begin{gather*}
\mathbf{M}_p(\mathbf{G}_p \boldsymbol\phi)= \int \mathbf{M}_p(\mathbf{z}_{p-1},d\mathbf{z}_p)
\mathbf{G}_{p}(\mathbf{z}_p)
\frac{1}{M_d}\sum_{l=1}^{M_d}\phi(z_{p,d+1}^{l,+}) =\\
\int \Big(\prod_{l=1}^{M_d}M_{p,1}(z_{p-1,d+1}^{l,+},dz_{p,1}^{l})\Big)
\Big(\prod_{j=2}^{d+1}\prod_{l=1}^{M_d}\Phi_{p,j}(\eta_{p,j-1}^{M_d})(dz_{p,j}^{l})\Big)\Big)
 \prod_{j=1}^d\eta_{p,j}^{M_d}(G_{p,j})\cdot
\eta_{p,d+1}^{M_d}(\phi).
\end{gather*}
So, the integral concerns now the local particle filter 
with weights $G_{p,j}$ and Markov kernels $M_{q,j}$. In particular, 
the integral corresponds to the expected value of the particle approximation 
of the standard Feynamn-Kac unnormalised estimator with standard unbiasedness properties  \cite[Theorem 7.4.2]{delmoral}. That is, the integral is equal to (here, for each $l$, the process
$z_{p,1}^l,z_{p,2}^l,\ldots, z_{p,d+1}^l$ is a Markov chain evolving 
via $M_{p,1}(z_{p-1,d+1}^{l,+},dz_{p,1}^l), M_{p,2}(z_{p,1}^l,dz_{p,2}^l),\ldots,$ $
M_{p,d+1}(z_{p,d}^l,dz_{p,d+1}^l)$ respectively) 
$$\frac{1}{M_d}\sum_{l=1}^{M_d} \mathbb{E}\big[ \phi(z_{p,d+1}^l) G_{p,d}(z_{p,d}^l) \cdots G_{p,2}(z_{p,2}^l)G_{p,1}(z_{p,1}^l)|z_{p-1,d+1}^{l,+} \big]. $$ 
From the analytical definition of the kernels and the weights, 
this latter quantity is easily seen to be equal to
\begin{gather*}
\frac{1}{M_d}\sum_{l=1}^{M_d}\int \phi(z) \prod_{j=1}^{d}\alpha_{p,j}(y_p,z_{p-1,d+1}^{l,+},z(1:j)) 
dz(1:j) = \frac{1}{M_d}\sum_{l=1}^{M_d} \int \phi(z) f_p(z_{p-1,d+1}^{l,+},dz)g_p(z,y_p)dz\\
= \eta_{p-1,d+1}^{M_d}(f_p(g_p\phi)).
\end{gather*}
So, we have obtained that
\begin{equation}
\label{eq:important}
\mathbf{M}_p(\mathbf{G}_p \boldsymbol\phi) = \eta_{p-1,d+1}^{M_d}(f_p(g_p\phi))
\in \mathcal{A}_{p-1}.
\end{equation}
Thus, applying the above result recursively, we obtain from (\ref{eq:ante}) that
\begin{equation}
\label{eq:aa}
\boldsymbol{\gamma}_n(\mathbf{G}_n\boldsymbol{\phi}) =\int \prod_{p=1}^{n}f_p(x_{p-1},dx_p)g_p(x_{p},y_p)\phi(x_p).
\end{equation}
Using the standard Feynman-Kac notation, this latter integral 
can be denoted as $\gamma_n(g_n\phi)$ for the unnormalised measure $\gamma_n$.
Thus, for instance, for the normalising constants, we have that
\begin{equation}
\label{eq:ae}
 \boldsymbol{\gamma}_{n}(\mathbf{G}_{n})  = \gamma_n(g_n) \equiv p(y_{1:n}).
\end{equation}
%

\subsection{Proof}
We have established that the algorithm is a standard particle filter approximation of a Feynman-Kac formula on an extended space. Thus, standard results, e.g.~in \cite{delmoral}, 
will give consistency for Monte-Carlo estimates on the enlarged 
state-space. In only remains to show that indeed the quantities in the statement
of  Theorem \ref{theo:consis} correspond to Monte-Carlo averages of the global 
filter in the enlarged space. 
We look directly at the last two quantities in the statement of the Theorem,
as the derivation for the first two ones is similar and simpler.
For the first we set 
\begin{equation*}
\boldsymbol{\varphi}(\textbf{z}_n) = 
\frac{1}{M_d}\sum_{l=1}^{M_d}\varphi(z_{n,d+1}^{l,+})\in \mathcal{A}_n,
\end{equation*}
and we immediately have that (denoting by $\check{\textbf{z}}_n^{i}$ the resampled 
islands, under the weights $\textbf{G}_n(\textbf{z}_n^{i})$)
\begin{equation*}
\frac{1}{N}\sum_{i=1}^{N}
\boldsymbol{\varphi}(\check{\textbf{z}}_n^{i}) \rightarrow_{\mathbb{P}} 
\frac{\int \boldsymbol{\eta}_n(d\textbf{z}_n)\textbf{G}_n(\textbf{z}_n)\boldsymbol{\varphi}(\textbf{z}_n)}
{\int \boldsymbol{\eta}_n(d\textbf{z}_n)\textbf{G}_n(\textbf{z}_n)} = 
\frac{\boldsymbol{\gamma}_n(\mathbf{G}_n\boldsymbol{\varphi})}{{\boldsymbol{\gamma}}_n(\mathbf{G}_n)}.
\end{equation*}
Notice now that the quantity on the left is precisely the double average 
in the statement of the Theorem and the quantity on  the right, from (\ref{eq:aa}),
is equal to $\gamma_n(g_n\varphi)/\gamma_n(g_n)=\pi_n(\varphi)$. 
For the last statement in the Theorem, the quantity on the left is 
$\boldsymbol{\gamma}_n^N(\textbf{G}_n)$ which, from standard particle 
filter theory converges in probability to $\boldsymbol{\gamma}_n(\textbf{G}_n)
 = \gamma_n(g_n) = p(y_{1:n})$.

\section{Monte Carlo Averages}\label{app:mc_avg}

Below let $V\in\mathbb{R}$ be a random variable with probability density $\alpha(x)/\int_{\mathbb{R}}\alpha(x)dx$. Recall that 
$\check{X}^{i,l}_n(d)$ is particle $i$, local particle $l$ at observation time $n$, dimension $d$ and it has just been locally resampled
using the weights $G_{n,d}(x_{n}^{i,l}(d))$. Recall that there is no global resampling. Throughout $M_d=d/c$ (assumed to be integer, for notational convenience). 

\begin{lem}\label{lem:asymp_ind}
Let $n\geq 1$, $i\in\{1,\dots,N\}$, $l\in\{1,\dots,M_d\}$ be fixed and $\varphi\in\mathcal{B}_b(\mathbb{R})$. Then 
\begin{enumerate}
\item{$\mathbf{G}_n(x_n^{i,1:M_d}(1:d))/(\int\alpha(x)dx)^d$ and
$ \varphi(\check{x}^{i,l}_n(d))$} 
\item{$\mathbf{G}_n(x_n^{i,1:M_d}(1:d))(\int\alpha(x)dx)^d$ and
$\frac{1}{M_d}\sum_{l=1}^{M_d}\varphi(\check{x}^{i,l}_n(d))$}
\end{enumerate}
are asymptotically independent as $d\rightarrow\infty$.
\end{lem}

\begin{proof}
We first consider statement 1. Set $r=\sqrt{-1}$ and consider the standardised quantity $\overline{\mathbf{G}}_n(x_n^{i,1:M_d}(1:d))=\mathbf{G}_n(x_n^{i,1:M_d}(1:d))/(\int\alpha(x)dx)^d$, then we have that for $(t_1,t_2)$ fixed,
\begin{gather*}
\mathbb{E}\Big[\exp\big\{rt_1\overline{\mathbf{G}}_n(X_n^{i,1:M_d}(1:d))+ rt_2 \varphi(\check{X}^{i,l}_n(d))\big\}\Big] = 
\\
\mathbb{E}\Bigg[\exp\big\{rt_1\overline{\mathbf{G}}_n(X_n^{i,1:M_d}(1:d))\big\} 
\frac{\sum_{l=1}^{M_d}G_{n,d}(X_{n}^{i,l}(d))e^{rt_2\varphi(X^{i,l}_n(d))}}{\sum_{l=1}^{M_d}G_{n,d}(X_{n}^{i,l}(d))}
\Bigg].
\label{eq:nono}
\end{gather*}
By standard SLLN, we have that
\begin{equation*}
\frac{\sum_{l=1}^{M_d}G_{n,d}(X_{n}^{i,l}(d))e^{rt_2\varphi(X^{i,l}_n(d))}}{\sum_{l=1}^{M_d}G_{n,d}(X_{n}^{i,l}(d))} \rightarrow_{\mathbb{P}} \frac{\int_{\mathbb{R}}\alpha(x)e^{rt_2\varphi(x)}dx}{\int_{\mathbb{R}}\alpha(x)dx}.
\end{equation*}
Also, Proposition \ref{prop:conv_weights} implies that
$$
\exp\{rt_1\overline{\mathbf{G}}_n(X_n^{i,1:M_d}(1:d))\} \Rightarrow \exp\{rt_1 V_n^i\}
$$
where $V_n^i\sim\mathcal{LN}(-c\sigma^2/2,c\sigma^2)$ for $\sigma^2$ defined therein. Hence, from Slutsky's lemmas we have 
$$
\exp\{rt_1\overline{\mathbf{G}}_n(X_n^{i,1:M_d}(1:d))\} 
\frac{\sum_{l=1}^{M_d}G_{n,d}(X_{n}^{i,l}(d))e^{rt_2\varphi(X^{i,l}_n(d))}}{\sum_{l=1}^{M_d}G_{n,d}(X_{n}^{i,l}(d))} \Rightarrow \exp\{rt_1 V_n^i\}
\frac{\int_{\mathbb{R}}\alpha(x)e^{rt_2\varphi(x)}dx}{\int_{\mathbb{R}}\alpha(x)dx}.
$$
The proof of 1.~is concluded on noting the boundedness of the functions.

For the proof of 2.~we have
\begin{gather}
\mathbb{E}\Big[e^{rt_1\overline{\mathbf{G}}_n(X_n^{i,1:M_d}(1:d)) + rt_2 \frac{1}{M_d}\sum_{l=1}^{M_d}\varphi(\check{X}^{i,l}_n(d))}\Big] =\nonumber \\
\mathbb{E}\Big[e^{rt_1\overline{\mathbf{G}}_n(X_n^{i,1:M_d}(1:d))}\big[e^{ rt_2 \frac{1}{M_d}\sum_{l=1}^{M_d}\varphi(\check{X}^{i,l}_n(d))} -e^{ rt_2\pi(\varphi)} \big]\Big]
+e^{ rt_2\pi(\varphi)}\mathbb{E}\big[e^{rt_1\overline{\mathbf{G}}_n(X_n^{i,1:M_d}(1:d))}\big]\nonumber \\
=: A_d + B_d 
\label{eq:an}
\end{gather}
where we have used the short-hand $\pi(\varphi)=\int_{\mathbb{R}}\alpha(x)\varphi(x)dx/\int_{\mathbb{R}}\alpha(x)dx$. 
From standard importance sampling and resampling results (see e.g.\@ \cite{rubin}), we have that 
$$
\frac{1}{M_d}\sum_{l=1}^{M_d} \varphi(\check{X}^{i,l}_n(d)) \rightarrow_{\mathbb{P}} \frac{\int_{\mathbb{R}}\alpha(x)\varphi(x)dx}{\int_{\mathbb{R}}\alpha(x)dx}.
$$
So, returning to (\ref{eq:an}), we have obtained that $A_d\rightarrow_{\mathbb{P}} 0$,
thus
$$
\lim_{d\rightarrow\infty}\mathbb{E}\Big[e^{rt_1\overline{\mathbf{G}}_n(X_n^{i,1:M_d}(1:d)) + rt_2 \frac{1}{M_d}\sum_{l=1}^{M_d}\varphi(\check{X}^{i,l}_n(d)) }\Big]
= \exp\{ rt_2\pi(\varphi)\}\mathbb{E}[\exp\{rt_1V_{n}^i\}]
$$
which concludes the proof of 2..
\end{proof}

\section{Proof of Theorem \ref{theo:nc_is_ok}}\label{app:prf_nc}

Recall the notation for the global filter from  Appendix \ref{sec:prf_consis}.
We define the semi-group
\begin{equation*}
\mathbf{Q}_{p+1}(\mathbf{z}_p,d\mathbf{z}_{p+1}) = \mathbf{G}_p(\mathbf{z}_p)\mathbf{M}_{p+1}(\mathbf{z}_p,d\mathbf{z}_{p+1})
\end{equation*}
and we also set
\begin{equation}
\mathbf{Q}_{p,n}(\varphi) = \int \mathbf{Q}_{p+1}(\mathbf{z}_p,d\mathbf{z}_{p+1})\times\cdots\times \mathbf{Q}_{n}(\mathbf{z}_{n-1},d\mathbf{z}_{n})\boldsymbol{\varphi}(\mathbf{z}_{n}).
\label{eq:Q_def}
\end{equation}
Recall from the main result in (\ref{eq:important}) in Appendix \ref{sec:prf_consis}, connecting the global with the local filter,  
that $\mathbf{M}_p(\mathbf{G}_p) = \eta_{p-1,d+1}^{M_d}(f_p(g_p))$,
and upon an iterative application of this result
\begin{equation}
\label{eq:mn}
\mathbf{Q}_{p,n}(1)  =  \mathbf{G}_p(\mathbf{z}_p)
\eta_{p,d+1}^{M_d}(
\hat{q}_{p+1,n-1}(1)).
\end{equation}
We also have that $\boldsymbol{\gamma}_n(1)=\gamma_n(1)=\gamma_p(g_p\hat{q}_{p+1,n-1}(1)) = \pi_p(\hat{q}_{p+1,n-1}(1))\gamma_p(g_p)$ and, finally,  that 
$ \gamma_p(g_p) = \pi_{p-1}(f_p(g_p))\gamma_p(1)$. 
Using all these expressions, simple calculations will give
\begin{align}
\sigma^2_n &=  \sum_{p=1}^n \frac{\boldsymbol{\gamma}_p(1)^2}{
\boldsymbol{\gamma}_n(1)^2}
\boldsymbol{\eta}_p\bigg(\Big(\mathbf{Q}_{p,n}(1)-\boldsymbol{\eta}_p(\mathbf{Q}_{p,n}(1))\Big)^2\bigg) \nonumber \\
 &=  \sum_{p=1}^n \boldsymbol{\eta}_p\bigg(\Big(  
\frac{\mathbf{G}_p(\mathbf{z}_p)}{\mathbf{M}_p(\mathbf{G}_p)}A_p - 1 
 \Big)^2\bigg)
\label{eq:bb}
\end{align}
where we have defined
\begin{equation*}
A_p = \frac{\eta_p^{M_d}(\hat{q}_{p+1,n-1}(1))}{\pi_p(\hat{q}_{p+1,n-1}(1))}
\cdot \frac{\eta_{p-1}^{M_d}(f_p(g_p))}{\pi_{p-1}(f_p(g_p))}.
\end{equation*}
The main thing to notice now, is that $\mathbf{G}_p(\mathbf{z}_p)/\textbf{M}_p(\mathbf{G}_p)$ corresponds to the standard estimate of the normalising constant 
for the $p$-th local filter divided with its expected value, and we can 
use standard results from the literature to control its second moment. 
Indeed,  by Assumptions (A\ref{hyp:pstruc}-\ref{hyp:semigroup}) and  \cite[Theorem 16.4.1]{delmoral1} (see Remark \ref{rem:nc_rev_var}), there exists $\tilde{c}<\infty$ (which does not depend on $p$ or $\mathbf{z}_p$) such that for any $d\geq 1$ and any $M_d\geq \tilde{c}d$
$$
\mathbf{M}_p\bigg(\Big(\frac{\mathbf{G}_p(\mathbf{z}_p)}{\mathbf{M}_p(\mathbf{G}_p)}-1\Big)^2\bigg) \leq \frac{\tilde{c}(2+e)d}{M_d},
$$
where the upper-bound only depends on $d$ via the term $d/M_d$.
Notice also that $f_p(g_p)\equiv 
\hat{q}_{p-1,p}(1)$, so by (A\ref{hyp:semigroup}) and Jensen's inequality 
(so that $M_d/\sum_{i=1}^{M_d} x_i
 \le \sum_{i=1}^{M_d} 
\frac{1}{x_i}/M_d$ for positive $x_i$), we have  
\begin{equation*}
0 \le A_p \le c^4.
\end{equation*}
Thus, returning in (\ref{eq:bb}), and using the last two equations, 
we get, starting with the $C_2$-inequality
\begin{align*}
\boldsymbol{\eta}_p\bigg(\Big(  
\frac{\mathbf{G}_p(\mathbf{z}_p)}{\mathbf{M}_p(\mathbf{G}_p)}A_p - 1 
 \Big)^2\bigg) &\le 
2 \boldsymbol{\eta}_p\bigg(\Big(  
\frac{\mathbf{G}_p(\mathbf{z}_p)}{\mathbf{M}_p(\mathbf{G}_p)} - 1 
 \Big)^2 c^8 \bigg) + 2\boldsymbol{\eta}_p\big((A_p - 1)^2\big) \\
 &= 2c^8 \boldsymbol{\gamma}_{p-1}\Big(\mathbf{G}_{p-1}\mathbf{M}_p\bigg(  
\frac{\mathbf{G}_p(\mathbf{z}_p)}{\mathbf{M}_p(\mathbf{G}_p)} - 1 
 \Big)^2 \bigg)/\boldsymbol{\gamma}_p(1)+ 2\boldsymbol{\eta}_p\big((A_p - 1)^2\big)\\
 & \le \frac{2c^8\tilde{c}(2+e)d}{M_d} + 2c^8.
\end{align*} 
From here, one can easily complete the proof and hence we conclude.

\begin{rem}\label{rem:nc_rev_var}
In the proof of Theorem \ref{theo:nc_is_ok} we have used \cite[Theorem 16.4.1]{delmoral1}. This is a result on the relative variance of the particle estimate of the normalizing constant,
and as stated in \cite{delmoral1} does not include a function, i.e.~an estimate of the form $\prod_{j=1}^d \eta_{p,j}^{M_d}(G_{p,j}) \eta_{p,j}^{M_d}(\varphi)$ for some $\varphi\in\mathcal{B}_b(\mathbb{R}^d)$.  Based on personal communication with Pierre Del Moral, \cite[Theorem 16.4.1]{delmoral1} can be extended to include a function, by modification of the potential functions and the use of the final formula  in \cite[pp.~484]{delmoral1}.
\end{rem}

\end{document}